\documentclass[%
 reprint,
groupedaddress,
 amsmath,amssymb,
 aps,
 pra,
]{revtex4-1}
\usepackage[T1]{fontenc}
\usepackage[final]{graphicx}
\usepackage{dcolumn}% Align table columns on decimal point
\usepackage{amsfonts}
\usepackage{amssymb}
\usepackage{pbox}
\usepackage{amsmath}

\DeclareMathOperator{\tr}{tr}

\DeclareMathOperator{\re}{re}
\DeclareMathOperator{\im}{im}

\DeclareMathOperator{\artanh}{artanh}

\usepackage{caption}
\usepackage{subcaption}
\usepackage{mathtools}
\usepackage{braket}

\usepackage{chngcntr}
\usepackage{multirow}
\usepackage{tabularx}
\usepackage[labelsep=period]{caption}
\date{}
\usepackage{graphicx}
\usepackage{enumerate}
\usepackage{calc}
\usepackage{indentfirst}
\usepackage{booktabs}
\usepackage{dsfont}
\usepackage{arydshln}
\usepackage{bm}
\usepackage{amsthm}
\usepackage{commath}
\usepackage{xcolor}

\theoremstyle{definition}
\newtheorem{definition}{Definition}
\newtheorem{theorem}{Theorem}
\newtheorem*{theorem*}{Theorem}

\newtheorem*{conjecture*}{Conjecture}

\newtheorem*{corollary*}{Corollary}
\usepackage{natbib}
\setcitestyle{square,numbers}
\usepackage{changepage}

\begin{document}

\title{\textbf{Stabilizing entanglement in two-mode Gaussian states}}
\author{Tomasz Linowski$^{1,2}$, Clemens Gneiting$^3$,  {\L}ukasz Rudnicki$^{1,2}$}

\affiliation{
    $^1$International Centre for Theory of Quantum Technologies, University of Gdansk, 80-308 Gda{\'n}sk, Poland \\
    $^2$Center for Theoretical Physics, Polish Academy of Sciences, 02-668 Warszawa, Poland\\
    $^3$Theoretical Quantum Physics Laboratory, RIKEN Cluster for Pioneering Research, Wako-shi, Saitama 351-0198, Japan
}

\date{\today}

\begin{abstract}
We analyze the stabilizability of entangled two-mode Gaussian states in three benchmark dissipative models: local damping, dissipators engineered to preserve two-mode squeezed states, and cascaded oscillators. In the first two models, we determine principal upper bounds on the stabilizable entanglement, while in the last model, arbitrary amounts of entanglement can be stabilized. All three models exhibit a tradeoff between state entanglement and purity in the entanglement maximizing limit. Our results are derived from the Hamiltonian-independent stabilizability conditions for Gaussian systems. Here, we sharpen these conditions with respect to their applicability.
\end{abstract}

\maketitle

\section{Introduction}\label{section:Introduction}
\addcontentsline{toc}{section}{Introduction}
Among the various non-classical aspects of quantum mechanics, the radically unintuitive way in which  systems can become correlated, a consequence of {quantum entanglement}% which Einstein considered a ``spooky action at a distance'' \cite{EPR}
, had been a subject of ongoing controversy. 
%In the early days of quantum mechanics, such nonlocal correlations appeared unphysical, triggering Einstein, Podolsky and Rosen \cite{EPR}, or EPR in short, to use quantum entanglement as an argument against quantum mechanics as a fundamentally correct description of reality. It was not before almost three decades later that the EPR paradox, as it is called today, was resolved by John Bell in 1964 \cite{Bell}.
Today, a century after its discovery, quantum entanglement has emerged as one of the most prolific resources of quantum mechanics and it continues to broaden our understanding of nature, with ideas as speculative as time emerging as an entanglement phenomenon being subject to experimental testing \cite{time_as_entanglement}. More than that, however, quantum entanglement has the potential to revolutionize not just the way we think about the world, but the world itself. In close relation to quantum coherence, it is the core property underlying novel technologies such as superdense coding \cite{superdense_coding}, quantum teleportation \cite{quantum_teleportation}, measurement precision beyond the classical limit \cite{measurement_heisenberg_limit} and others \cite{quantum_information_book,quantum_algorithms}.

What often hinders us from harnessing entanglement is decoherence, i.e., the loss of quantum coherence, which tends to rapidly deteriorate the aforementioned quantum benefits in systems subject to even the mildest forms of interaction with an environment -- which in practice is usually inevitable. In the theory of quantum open systems, the influence of the environment on a system is often modeled by a Lindblad master equation \cite{GKS_original,lindblad_original,ChruscinskiReview}:
\begin{equation} \label{eq:lindblad}
\frac{d\hat{\rho}}{dt} = 
	-\frac{i}{\hbar} \big[\hat{H},\hat{\rho}\big] + \hat{D}(\hat{\rho}),
\end{equation}
where $\hat{H}$ is the system Hamiltonian and the dissipator $\hat{D}(\hat{\rho})$ encodes the effects of interaction with the environment (the detailed structure of the dissipator is explained below).

{Over the years, a variety of methods has been developed to deal with the presence of an environment. For example, the dissipative part of the dynamics can be employed to strengthen the desired features of the system \cite{using_dissipation_1,using_dissipation_2,Lyapunov_stationarity}.} For a fixed dissipator, the Hamiltonian remains as the only resource for {stabilizing} desired system states \cite{using_hamiltonian_1,using_hamiltonian_2}. The task is then to look for an appropriate control Hamiltonian $\hat{H}$, such that, for a given environment $\hat{D}(\hat{\rho})$, a desired state $\hat{\rho}$ becomes {\emph{stationary}}, that is, it is a solution to the Lindblad equation with vanishing left hand side.

A more general, geometric perspective has recently been taken in \cite{stabilizability_geometric, stabilizability_cv_systems}. Instead of on {stationary} states, here the focus lies on {\emph{stabilizable}} states, i.e., states, for which, given an environmental effect $\hat{D}(\hat{\rho})$, there exists an (unspecified) Hamiltonian $\hat{H}$, such that the aforementioned equation holds (in other words, stabilizable states may be regarded as families of {potentially stationary} states).

Here, we apply the theory of stabilizability to {two-mode Gaussian states}, that is, bipartite continuous-variable states with {normally distributed} Wigner function. Gaussian states are among the most generic, yet most useful states both in theoretical and in experimental quantum optics \cite{gaussian_optics}, as well as quantum information \cite{gaussian_information_1,gaussian_information_2,Lyapunov_stationarity,controllabel_generation_of_two_mode_entangled_states}. They include, among others, coherent, squeezed, and thermal states \cite{serafini}. In particular, with regard to the importance of entanglement as a resource, we investigate which entangled states can be stabilized and what is the maximum amount of entanglement admitted within the set of stabilizable states.

{We focus here on two-mode Gaussian states, due to their fundamental importance in many quantum information protocols. The entanglement properties of two-mode Gaussian states and their experimental feasibility are thoroughly studied and well understood \cite{controllabel_generation_of_two_mode_entangled_states,Gaussian_states_in_technologies,two-mode_gaussian_etc_proper_norm,Marian_2001}. While the formalism of stabilizable states is readily applicable to more than two modes \cite{stabilizability_cv_systems}, multipartite scenarios do in general not admit a unified treatment of their entanglement properties \cite{Three_qubits_can_be_entangled_in_two_inequivalent_ways,entangling_power_of_multipartite_gates,geometry_of_quantum_states}.}

We consider the stabilizability of entangled Gaussian states within three paradigmatic dissipative models of two-mode systems: two modes subject to local damping, dissipators engineered to preserve two-mode squeezed thermal states and cascaded oscillators coupled to the vacuum \cite{stabilizability_cv_systems, two_modes_damping}. All three models have found use in the context of quantum technologies, ranging from quantum cryptography and computation \cite{Gaussian_states_in_technologies,OPOs_computation}, to experimental generation of entanglement \cite{OPOs_entanglement,squeezed_dissipation_experimental_1,squeezed_dissipation_experimental_2}, to spectroscopy \cite{OPOs_spectroscopy}, among others. Moreover, these models have been the focus of recent theoretical investigations, see, e.g., \cite{two_modes_damping,Lyapunov_stationarity}. Finally, which is not without importance for our purposes, the models can, to a large extent, be treated analytically, giving deeper insights into the mechanism in question.

In the case of local damping, where the dissipator clearly acts adversary to entanglement, our findings give evidence that the amount of entanglement achievable within the set of stabilizable states is upper bounded by $\log 2$, as quantified by logarithmic negativity. Surprisingly, we find that a similar upper bound also exists for dissipators engineered to preserve two-mode squeezed thermal states, i.e., dissipators which are fundamentally nonlocal. On the other hand, we prove that it is possible to stabilize states that are more entangled than the two-mode squeezed states underlying the engineered dissipator. In the remaining model of the cascaded oscillators, we show that, in principle, arbitrary amount of entanglement can be stabilized. In all three cases we observe that the stabilizable states characterized by the maximum amount of entanglement are close to be maximally mixed, suggesting an asymptotic tradeoff relation between entanglement and purity within the stabilizable states. This is reminiscent of previous findings \cite{stabilizability_geometric} regarding two qubits.

This work is organized as follows: In Section \ref{sec:Preliminaries} we briefly summarize the main characteristics of (two-mode) Gaussian states, along with our chosen measures of entanglement and mixedness. In Section \ref{section:Stabilizability} we rigorously introduce the notion of stabilizability and prove Theorem~\ref{theorem_covariance}, in which we sharpen the necessary conditions for stabilizability of general Gaussian states derived previously \cite{stabilizability_cv_systems}, showing that half of these conditions are always automatically fulfilled. Section \ref{sec:Results} is dedicated to our main results: stabilizability of two-mode entangled states in the three considered environmental models. Finally, in Section \ref{sec:Summary}, we discuss our results {and their limitations, as well as possible generalizations,} and give an outlook for future research.

\section{Gaussian states}\label{sec:Preliminaries}
Let us consider an $N$-mode Hilbert space $\mathcal{H} = \bigotimes_{i=1}^N\mathcal{H}_i$ described by the vector of $N$ pairs of position and momentum operators 
\begin{equation} \label{xi}
\hat{\vec{\xi}} \coloneqq (\hat{x}_1, \hat{p}_1, \ldots, \hat{x}_N, \hat{p}_N)^T.
\end{equation}
The canonical commutation relations
\begin{equation} \label{canonical_commutation_relations}
\big[\hat{x}_j,\hat{p}_k\big]=i\hbar\delta_{jk},
    \qquad \big[\hat{x}_j,\hat{x}_k\big]
    =\big[\hat{p}_j,\hat{p}_k\big]=0,
\end{equation}
can be concisely encoded in the so-called {\emph{symplectic form}}
\begin{equation} \label{symplectic_form}
\begin{split}
J_{jk} \coloneqq -\frac{i}{\hbar}\big[\hat{\xi}_j,\hat{\xi}_{k}\big],
\end{split}
\end{equation}
which explicitly reads
\begin{equation}
J=\bigoplus_{k=1}^N
\begin{bmatrix}
0&1\\
-1&0
\end{bmatrix}.
\end{equation}

Following standard terminology we call {\emph{Gaussian states}} all the states with normal (Gaussian) characteristic functions and quasiprobability distributions \cite{using_dissipation_2,stabilizability_cv_systems,cv_systems_gaussian_states,two-mode_gaussian_etc,two-mode_gaussian_etc_proper_norm}. It follows from this definition that Gaussian states are fully characterized by the first and second moments of the vector $\hat{\vec{\xi}}$. The first moments can be adjusted to have an arbitrary value with local operations, which do not affect global properties of the state such as entanglement or mixedness, and can thus be set to 0. Therefore, from the point of view of this work, any Gaussian state is {fully described} by the set of second moments of the vector $\hat{\vec{\xi}}$, conveniently encoded in the {\emph{covariance matrix}}
\begin{equation} \label{covariance_matrix}
V_{kl}=V_{lk} \coloneqq \frac{1}{2}\braket{\big\{\hat{\xi}_k, \hat{\xi}_{l}\big\}},
\end{equation}
where $\{\cdot,\cdot\}$ denotes the anticommutator. 

In the particular case of two-mode Gaussian states, $N=2$, any valid covariance matrix possesses a simple, unique form, called the {\emph{standard form}} \cite{two-mode_gaussian_etc,two-mode_gaussian_etc_proper_norm}:
\begin{equation} \label{standard}
V_{\textrm{sf}} = 
\begin{bmatrix}
a & 0 & c_+ & 0 \\
0 & a & 0 & c_- \\
c_+ & 0 & b & 0 \\
0 & c_- & 0 & b
\end{bmatrix},
\end{equation}
where the parameters $a, b > 0$ are proportional to the average number of particles / excitations in the two modes and the coefficients $c_\pm \in \mathbb{R}$ contain the information about the correlations between the modes. Any two-mode covariance matrix can be brought into its standard form by means of local symplectic operations, which, similarly to local unitary operations for density matrices, do not change global properties of the state. For this reason, unless stated otherwise, from now on we assume $V$ to be in its standard form.

Note that not all matrices (\ref{standard}) constitute valid covariance matrices of two-mode Gaussian states. For this to be the case, they need to additionally fulfill the Heisenberg uncertainty principle:
\begin{equation} \label{uncertainty_principle}
\sqrt{\braket{\hat{x}_k^2}-\braket{\hat{x}_k}^2}
	\sqrt{\braket{\hat{p}_k^2}-\braket{\hat{p}_k}^2}
	\geqslant \hbar/2,
\end{equation}
where $k\in\{1,2\}$, equivalent to \cite{two-mode_gaussian_etc_proper_norm}
\begin{equation} \label{Heis}
\begin{split}
2 \leqslant 4 \Delta(V) \leqslant 1 + 16 \det V,
\end{split}
\end{equation}
with $\Delta(V) \coloneqq a^2 + b^2 + 2 c_+c_-$ and 
\begin{equation} \label{detV}
\begin{split}
\det V = \big(ab-c_+^2\big)\big(ab-c_-^2\big),\quad
    ab - c_\pm^2 \geqslant 0.
\end{split}
\end{equation}

Since it will become relevant below, we remark that the parametrization of the standard form (\ref{standard}) in terms of $(a,b,c_\pm)$ is not the only valid choice. Of particular significance is also the description in terms of the {\emph{symplectic eigenvalues}} of~$V$:
\begin{equation} \label{symplectic_eigenvalues}
\begin{split}
1/2 \leqslant \nu_- \leqslant \nu_+.
\end{split}
\end{equation}
The symplectic eigenvalues are the eigenvalues of the matrix product $JV$ and read explicitly
\begin{equation} \label{nu_+-}
\begin{split}
\nu_\pm(V) = 
\sqrt{\frac{1}{2}\left(\Delta(V) \pm \sqrt{\Delta^2(V) - 4\det V}\right)}.
\end{split}
\end{equation}
An important subclass of two-mode Gaussian states, that is most easily described in terms of symplectic eigenvalues, consists of {\emph{nonsymmetric two-mode squeezed thermal states}}
\begin{equation}
\begin{split}
\hat{\rho}_{\textnormal{sq}}(\nu_\pm,r) = \hat{S}(r)\hat{\rho}_{\textnormal{th}}(\nu_\pm)\hat{S}^\dag(r),
\end{split}
\end{equation}
which arise from applying the two-mode squeezing operator
\begin{equation} \label{squeezing_operator}
\begin{split}
\hat{S}(r) = e^{\frac{r}{2}
	\left(\hat{a}_1 \hat{a}_2 - \hat{a}_1^\dag \hat{a}_2^\dag \right)},
\end{split}
\end{equation}
where $\hat{a}_k$ is the annihilation operator of the $k$-th mode and $r$ is the squeezing parameter, to the two-mode thermal state $\hat{\rho}_{\textnormal{th}}(\nu_\pm)$. Most importantly, squeezed states of light are used in quantum metrology: as means of enhancing the measurement precision \cite{squeezed_metrology}. For a detailed review, see \cite{squeezed_review}.

It can be shown that the standard form of the covariance matrix for such states reads \cite{two-mode_gaussian_etc}:
\begin{equation} \label{squeezed_parametrization}
\begin{split}
a(\nu_\pm, r) &= \nu_- \cosh^2 r + \nu_+ \sinh^2 r,\\
b(\nu_\pm, r) &= \nu_- \sinh^2 r + \nu_+ \cosh^2 r,\\
c_\pm(\nu_\pm, r) &= \pm \frac{\nu_- + \nu_+}{2}\sinh 2r,
\end{split}
\end{equation}
which we will refer to as the {\emph{squeezed state parametrization}} in the remainder. In fact, every state that fulfills $ c_+ = -c_- $ and $a,b\geqslant 1/2$ can be parametrized using the above recipe. The former requirement is obvious, while the latter arises from the fact that
\begin{equation}
\begin{split}
    a(\nu_\pm, r)\geqslant\nu_- (\cosh^2 r + \sinh^2 r) 
        \geqslant \nu_- \geqslant 1/2,
\end{split}
\end{equation}
and analogously for $b$. It is easy to show that any two-mode squeezed state is physical, that is, fulfills the Heisenberg uncertainty relation (\ref{Heis}).

\subsection{Entanglement measure}
Since we are interested in stabilizing entangled states, we need a way to certify entanglement. For two-mode Gaussian states, a necessary and sufficient separability criterion is given by the extension of the {PPT criterion} \cite{PPT} to continuous variable systems \cite{PPT_cv_systems}. This criterion states that, if the partial transposition of the state with respect to a given bipartition is not positive semi-definite, then the state is entangled with respect to this bipartition.

For two-mode Gaussian states in the covariance matrix representation, partial transposition with respect to the second mode corresponds to a mirror reflection of the second momentum: $p_2 \to - p_2$. This changes the symplectic eigenvalues of the state from (\ref{nu_+-}) to
\begin{equation} \label{nu_+-tilde}
\begin{split}
\tilde{\nu}_\pm(V) = 
	\sqrt{\frac{1}{2}\left(\tilde{\Delta}(V)\pm\sqrt{\tilde{\Delta}^2(V)-4\det V}\right)},
\end{split}
\end{equation}
where $\tilde{\Delta}(V) \coloneqq a^2 + b^2 - 2 c_+ c_-$. The PPT criterion thus reads \cite{two-mode_gaussian_etc_proper_norm}
\begin{equation}
\begin{split}
\tilde{\nu}_-(V) \geqslant 1/2,
\end{split}
\end{equation}
since $\tilde{\nu}_-(V) < 1/2$ would result in an invalid covariance matrix [see eq. (\ref{symplectic_eigenvalues})]. We stress that, in the case of two-mode Gaussian states, the PPT criterion is both necessary and sufficient \cite{PPT_cv_systems}.

We now have a simple way of certifying the presence of entanglement in Gaussian states. However, we still need a way to quantify it. {Several different measures of entanglement of two-mode Gaussian states have been proposed, including entanglement of formation, Bures distance, and Gaussian measures of entanglement \cite{Gaussian_entanglement_measures,Bures_distance}.} In this work, we deploy the {\emph{logarithmic negativity}}, defined as
\begin{equation}
\begin{split}
E_\mathcal{N}(\hat{\rho}) \coloneqq \log\tr\big|\hat{\rho}^{T_2}\big|,
\end{split}
\end{equation}
where $\hat{\rho}^{T_2}$ is the partially transposed state. The logarithmic negativity constitutes an upper bound to the {distillable entanglement} in the state, and it is continuous, convex and monotone under local operations and classical communication as long as the considered state has a finite mean energy. In other words, it is a proper measure of entanglement.

In the case of two-mode Gaussian states, the logarithmic negativity takes a particularly simple form \cite{two-mode_gaussian_etc_proper_norm}:
\begin{equation} \label{log_neg}
\begin{split}
E_\mathcal{N}(V) \coloneqq \max \left\{0,-\log\left[2\tilde{\nu}_-(V)\right]\right\}.
\end{split}
\end{equation}
%The logarithmic negativity will be our main means for entanglement quantification in Gaussian states throughout the rest of the article.

\subsection{Measures of mixedness}
As shown below, the amount of entanglement in stabilizable states is related to their purity. In order to verify this, in addition to the degree of entanglement of stabilizable states, we also need to characterize their degree of purity. It is known that any pure state's covariance matrix fulfills $a=b$, and $c_+=-c_-=\sqrt{a^2-1/4}$. However, to cover the general case we also need to select measures of mixedness.

The purity of the state is defined as $\mu(\hat{\rho})\coloneqq \tr\hat{\rho}^2$. For our purposes, it is more convenient to consider the degree of mixedness being state's {lack of purity}. One of the most often used measures of mixedness is given by the {\emph{linear entropy}}
\begin{equation} \label{eq:S_L}
\begin{split}
S_L(\hat{\rho})\coloneqq 1 - \mu(\hat{\rho}),
\end{split}
\end{equation}
which is essentially a linearized version of the von Neumann entropy $S_V(\hat{\rho})\coloneqq-\tr\hat{\rho}\log\hat{\rho}$. Both entropies are just special cases of the {Tsallis} \cite{entropies_tsallis} and {R\'{e}nyi entropies} \cite{entropies_renyi}. 

For two-mode Gaussian states~\cite{two-mode_gaussian_etc} we can calculate
\begin{equation} \label{trrho}
\begin{split}
\tr\hat{\rho}^p&=g_p(2\nu_+)g_p(2\nu_-),\\
g_p(x)&\coloneqq2^p[(x+1)^p-(x-1)^p]^{-1}.
\end{split}
\end{equation}
In particular, the linear entropy (\ref{eq:S_L}) reduces to the simple expression
\begin{equation} \label{S_L_simple}
\begin{split}
S_L(V) = 1-\left[4\nu_+(V)\nu_-(V)\right]^{-1}.
\end{split}
\end{equation}
Due to its relative simplicity, throughout the rest of this work, the linear entropy will be our choice for the measure of mixedness. However, we numerically obtain qualitatively similar results for some of the other measures mentioned above.

\section{Stabilizability} 
\label{section:Stabilizability}
We complete our toolbox by introducing the conditions for {stabilizability}. Let us start with general states $\hat{\rho}$ evolving under the GKLS (or Lindblad in short) equation (\ref{eq:lindblad}). The dissipator has the form 
\begin{equation} \label{dissipator}
\begin{split}
\hat{D}\left(\hat{\rho}\right)\coloneqq
	\sum_{k} \left( \hat{L}_k \hat{\rho} \hat{L}_k^\dag
	-\frac{1}{2}\big\{\hat{L}_k^\dag \hat{L}_k,\hat{\rho}\big\}\right),
\end{split}
\end{equation}
where $\hat{L}_k$ are the so-called {\emph{Lindblad operators}}.

In \cite{stabilizability_cv_systems}, the following two definitions were distinguished:

\begin{definition}
A state $\hat{\rho}$ is a \emph{stationary} state of the Lindblad equation (\ref{eq:lindblad}), if $d\hat{\rho}/dt=0$.
\end{definition}

\begin{definition}
A state $\hat{\rho}$ is a \emph{stabilizable} state with respect to the dissipator $\hat{D}(\hat{\rho})$, if there exists a Hamiltonian $\hat{H}$ such that $\hat{\rho}$ is the stationary state of the Lindblad equation (\ref{eq:lindblad}) with this specific Hamiltonian as an input.
\end{definition}

Both definitions are concerned with robustness of the system against the action of the environment. However, while stationarity is formulated with respect to both the Hamiltonian and the dissipator, stabilizability refers only to the latter. Consequently, it follows \cite{stabilizability_geometric} that the set of stabilizable states with respect to the dissipator $\hat{D}(\hat{\rho})$,
\begin{equation}
S_{\hat{D}} \coloneqq \big\{\hat{\rho}:\exists{\hat{H}} \:\:
	0 = -\frac{i}{\hbar} [\hat{H}, \hat{\rho}] + \hat{D}(\hat{\rho})\big\},
\end{equation}
is independent of the Hamiltonian.

We note in passing that, by definition, any stationary state is necessarily stabilizable. Thus, by considering stabilizability, we can make meaningful statements about whether a given state has the potential to be a stationary solution to the Lindblad equation {without the need to specify a Hamiltonian}.

In \cite{stabilizability_geometric}, the following {necessary} conditions for stabilizability of general (finite-dimensional) quantum systems were derived: a state $\hat{\rho}$ is stabilizable, if
\begin{equation} \label{stabilizability_general}
0=\tr\left[\hat{\rho}^{k}\hat{D}(\hat{\rho})\right],
\end{equation}
for $k\in\{1,\ldots,d-1\}$, where $d$ denotes the dimension of the Hilbert space. These conditions are based on the insight that, at stationarity, the Hamiltonian must be able to compensate/neutralize the effect of the dissipator, which implies that the dissipator must not affect the moments of the state, such as the purity.

\subsection{Stabilizability of Gaussian states}

In the context of continuous variable systems, including Gaussian states, the general stabilizability conditions~(\ref{stabilizability_general}) cannot be applied directly, since one must in general check infinitely many conditions. More importantly, however, the general conditions (\ref{stabilizability_general}) leave the Hamiltonian unconstrained. While this allows for considerations of most general nature, in many situations natural constraints limit the range of accessible Hamiltonians. This is especially the case in experiments, which are often, due to technical limitations, restricted to {quadratic Hamiltonians} that are at most quadratic in the creation and annihilation operators. In particular, the structure-preserving evolution of Gaussian states is driven by such quadratic Hamiltonians.

For this reason, a different methodology, incorporating this constraint, has recently been developed \cite{stabilizability_cv_systems}. In the case of quadratic Hamiltonians, i.e., Hamiltonians of the form:
\begin{equation} \label{eq:hamiltonian_quadratic}
\begin{split}
\hat{H} = \hat{\vec{\xi}}^T G \hat{\vec{\xi}},
\end{split}
\end{equation}
where $G$ is a $2N\times 2N$, real, symmetric matrix and $\hat{\vec{\xi}}$ is the vector of mode quadratures defined by Eq. (\ref{xi}), the Lindblad evolution of the covariance matrix (of any state, not necessarily Gaussian) can be concisely written as \cite{using_dissipation_1,using_dissipation_2}
\begin{equation} \label{eq:v_evolution}
\frac{d}{dt}V = AV+VA^T+J (\re C^\dag C)J^T.
\end{equation}
The matrix $A\coloneqq J\left[G+(\im C^\dag C)\right]$ is not symmetric in general, while
\begin{equation} \label{eq:matrix_C}
C_{kl}\coloneqq (\vec{c}_k)_{l}
\end{equation}
is a $2N\times 2N$ matrix resulting from writing the Lindblad operators as $\hat{L}_k=\vec{c}_k \cdot \hat{\vec{\xi}}$ with $\vec{c}_k\in\mathbb{C}^{2N}$. It is assumed that the Lindblad operators are linear in $\hat{x}_k$, $\hat{p}_k$ in order to guarantee consistency with the quadratic nature of the time evolution.

It has been shown \cite{stabilizability_cv_systems} that the necessary conditions for stabilizability of the covariance matrix read
\begin{equation} \label{stabilizability}
0 = 2 \tr \left(I_C J \tilde{V}^k \right) + 
	\tr \left(R_C J \tilde{V}^{k-1} \right),
\end{equation}
where $k\in\{1,\ldots,2N\}$, and we have introduced the~short-hand notation $I_C \coloneqq \im C^\dag C$, $R_C \coloneqq \re C^\dag C$, $\tilde{V}\coloneqq J V$.

We now prove that, for all odd $k$, the conditions (\ref{stabilizability}) are automatically satisfied. This will considerably simplify our analysis of two-mode Gaussian states below.

\begin{theorem}\label{theorem_covariance}
Let $l\in\mathbb{N}$. Then for all $V$, $C$ as in (\ref{stabilizability})
\begin{equation} \label{eq:theorem_supp_cv}
2 \tr \left[I_C J \tilde{V}^{2l+1} \right] + 
	\tr \left[R_C J \tilde{V}^{2l} \right] = 0.
\end{equation}
\end{theorem}
 
\begin{proof}
Let us denote the first trace by $X$. Since transposition does not change the value of the trace, we have
\begin{equation}
\begin{split}
X = \tr \left[ I_C J \tilde{V}^{2l+1} \right]^T	= \tr \left[(V^T J^T)^{2l+1} J^T I_C^T \right].
\end{split}
\end{equation}
The matrices $J$, $C$ satisfy  $J^T=-J$, $J^2=-\mathds{1}_{2N}$, $R_C^T=R_C$ and $I_C^T=-I_C$ \cite{stabilizability_cv_systems}. Performing all the transpositions accordingly produces an extra minus sign:
\begin{equation}
\begin{split}
X = - \tr \left[(V J)^{2l+1} J I_C \right].
\end{split}
\end{equation}
We can now use the fact that $J^2 = - \mathds{1}_{2N}$ to~cancel out the last two $J$ matrices. At~the same time, we can insert $\mathds{1}_{2N}=-J^2$ in front of~the trace. Obviously, this produces no overall change in~sign:
\begin{equation}
\begin{split}
X = - \tr \left[J(J V)^{2l+1}I_C \right] =- \tr \left[I_C J \tilde{V}^{2l+1} \right] = -X,
\end{split}
\end{equation}
where we have used the cyclic property of~the trace. Therefore, we have shown that the first term in~(\ref{eq:theorem_supp_cv}) equals its negative, and thus vanishes for all $l$. The second term vanishes in~an analogous way.
\end{proof}

{Theorem~\ref{theorem_covariance}} states that all the odd ($k\in\{1,3,\ldots\}$) stabilizability conditions (\ref{stabilizability}) are always fulfilled. Thus, in order to investigate the stabilizability of an $N$-mode covariance matrix, one needs to solve only $N$ rather than $2N$ equations.

\section{Stabilizability of entangled two-mode Gaussian states}
\label{sec:Results}
The reduced number of stabilizability conditions (applying Theorem \ref{theorem_covariance}) allows us to investigate the stabilizability of two-mode entangled states analytically. If we denote by $\vec{z} \coloneqq (a,b,c_+,c_-)$ the set of variables parametrizing the state $V$, and by $\vec{t}$ the additional parameters that come from the dissipator (\ref{dissipator}) [and thus parametrize the matrix $C$ in Eq. (\ref{stabilizability})], then the desired covariance matrices $V(\vec{z})$ describe states that are \begin{subequations}
\begin{enumerate}[(i)]
    \item entangled -- that is, are characterized by positive logarithmic negativity: 
    \begin{equation} \label{constraints_general_log_neg}
    E_\mathcal{N}(\vec{z})>0,
    \end{equation}
    
    \item physical -- that is, satisfy the Heisenberg uncertainty principle (\ref{Heis}):
    \begin{equation} \label{constraints_general_h}
    \begin{split}
        h_{1}(\vec{z}) & \coloneqq 
        	4\Delta(\vec{z})-16\det V(\vec{z})-1 \leqslant 0, \\
        h_{2}(\vec{z}) & \coloneqq 
        	-4\Delta(\vec{z})+2 \leqslant 0,
    \end{split}
    \end{equation}
    
    \item stabilizable -- that is, satisfy the conditions (\ref{stabilizability}) for $k=2$ and $k=4$:
    \begin{equation} \label{constraints_general_g}
    \begin{split}
        g_{1}(\vec{z},\vec{t}) &\coloneqq
        	2 \tr \left[I_C(\vec{t}) 
    		J \tilde{V}^2(\vec{z}) \right] 
    		+ \tr \left[R_C(\vec{t}) 
    		J \tilde{V}(\vec{z}) \right] = 0, \\
        g_{2}(\vec{z},\vec{t}) &\coloneqq
        	2 \tr \left[I_C(\vec{t}) 
    		J \tilde{V}^4(\vec{z}) \right]
    		+ \tr \left[R_C(\vec{t}) 
    		J \tilde{V}^3(\vec{z}) \right] = 0 .
    \end{split}
    \end{equation}  
\end{enumerate}\end{subequations}
The existence and specific form of the solutions to the equation system (i)--(iii) depend on the dissipative model at hand. {We emphasize that, while bath engineering may introduce some flexibility on the side of the dissipator \cite{using_dissipation_1,using_dissipation_2,Lyapunov_stationarity}, we focus here on fixed dissipators (due to an uncontrolled bath, and/or due to an engineered, but fixed, environment).} This implies that, in our considerations, we generally treat the parameters $\vec{t}$ as fixed, while manipulating the vector~$\vec{z}$. 
%{That said, our results apply also if the environment is engineered. In each case, we analyze how the measures of entanglement and entropy change with respect to $\vec{t}$ in addition to $\vec{z}$. Furthermore, when discussing nonlocal baths, we consider the presence of local noise in the dissipators, relevant in scenarios in which the nonlocal bath is (imperfectly) engineered.}

{Interestingly,  while the stabilizability conditions (\ref{stabilizability}) are in general only necessary, one can, for all cases discussed below, determine the corresponding stabilizing Hamiltonians by solving eq. (\ref{eq:v_evolution}) with vanishing left-hand side. Thus, in all the cases discussed below, we can consider states satisfying the constraints (\ref{constraints_general_g}) to be stabilizable.}

\subsection{Two modes with local damping}
\label{sec:two_modes_with_local_damping}
In the case of {local damping}, the two modes interact with independent environments, resulting in uncorrelated loss of particles/excitations in the modes \cite{two_modes_damping}. This situation describes a generic challenge faced by technologies employing entanglement of two-mode Gaussian states, such as teleportation, quantum cryptography and quantum computation \cite{Gaussian_states_in_technologies}. Clearly, the local dissipators act adversarial to nonlocal resources such as entanglement. Therefore, it is relevant to analyze the amount of entanglement that can be upheld by appropriate choice of the control Hamiltonian.

The Lindblad operators have the form \cite{stabilizability_cv_systems}
\begin{equation} \label{lindblad_op}
\begin{split}
\hat{L}_k & \coloneqq \sqrt{\frac{\gamma_k}{2}} 
	\left( \frac{\hat{x}_k}{x_0} + i x_0 \hat{p}_k \right),
\end{split}
\end{equation}
where in the adopted notation the rates $\gamma_k \geqslant 0$, $k\in\{1,2\}$, are responsible for the strength of dissipation in each mode, and $x_0 \in \mathbb{R}_+$. Note that, if $x_0=1$, the operators (\ref{lindblad_op}) are proportional to the annihilation operators $\hat{a}_k\coloneqq(\hat{x}_k+i\hat{p}_k)/\sqrt{2}$ of the respective modes. In general, $x_0$ can be interpreted as the system's characteristic length scale, which, in the case of the standard harmonic oscillator, is determined by the Hamiltonian \cite{stabilizability_cv_systems}. Recall that, in our geometric approach, the Hamiltonian is a priori unknown; however, in principle it can always be determined \cite{stabilizability_geometric}.

We stress that, because the two modes interact with independent environments, in the absence of a control Hamiltonian, the steady state of the system (if it exists) is separable [this can be explicitly seen by setting $G=dV/dt=0$ in Eq. (\ref{eq:v_evolution})]. This reconfirms that the family of dissipators at hand is adversary to entanglement. 

The choice (\ref{lindblad_op}) implies
\begin{equation}
\begin{split}
\vec{c}_1(\vec{t}) &= \sqrt{\frac{\gamma_1}{2}} 
	\left( x_0^{-1}, i x_0, 0, 0 \right)^T, \\
\vec{c}_2(\vec{t}) &= \sqrt{\frac{\gamma_2}{2}} 
	\left( 0, 0, x_0^{-1}, i x_0\right)^T,
\end{split}
\end{equation}
where the parameters are $\vec{t}=(x_0,\gamma_1,\gamma_2)$. Substituting the resulting $C$ into (\ref{stabilizability}) [with $V$ taken in the standard form (\ref{standard})] then yields:
\begin{equation} \label{damped_cond}
\begin{split}
0=& \: g_{1}(\vec{z},\vec{t}) = 
	\frac{\gamma_1}{2} \left[\left(x_0^{-2}+x_0^{2}\right) a - 4 a^2\right] \\
	& \qquad\: + \frac{\gamma_2}{2} \left[\left(x_0^{-2}+x_0^{2}\right) b - 4 b^2\right] 
	- 2 (\gamma_1 + \gamma_2) c_+ c_-, \\
0=& \: g_{2}(\vec{z},\vec{t}) =  - 2(\gamma_1+\gamma_2)\left(ab-c_+^2\right)\left(ab-c_-^2\right) \\
    & + \frac{1}{2}(\gamma_2 a + \gamma_1 b)
	\left[ \left(x_0^{-2}+x_0^{2}\right) ab 
	- \left(x_0^{-2}c_+^2 + x_0^{2} c_-^2\right) \right],
\end{split}
\end{equation}
where we have simplified $g_{2}(\vec{z},\vec{t})$ assuming $g_{1}(\vec{z},\vec{t})=0$.

The above system can be solved, for example, by extracting $c_{+}(a,b,c_-,\vec{t})$ from the first equation, substituting it into the second equation, and then solving the second equation for $\big(c_{-}^2\big)_k(a,b,\vec{t})$, $k\in\{1,2\}$. The solution can then be inserted into the constraints (\ref{constraints_general_log_neg}, \ref{constraints_general_h}), yielding a rather complex set of inequalities{, see Appendix~\ref{sec:explicit_solution_to_the_general_damping_problems} for details}.

While this set of inequalities can still be solved numerically, we focus here on two special classes of states, for which we give exact solutions. Based on these solutions, we then argue about the expected results in the general case. The respective special cases concern states with standard form $c_+ = -c_- \equiv c$ \footnote{Note that it is necessary for an entangled state to have $c_+$ and $c_-$ with opposite signs. This can be shown by manipulating eq. (\ref{constraints_general_log_neg}) and taking into account the first of the constrains (\ref{constraints_general_h}).}, and states with standard form $a = b$.

We emphasize that both restrictions are natural, with the former in particular being fulfilled by all squeezed thermal states. In both cases, we show that the maximum value of logarithmic negativity cannot exceed $E_{\mathcal{N},\max}=\log 2$, and that this value is obtained only, or most {easily} (as explained below), if $\gamma_1=\gamma_2$ and $x_0=1$. We then argue that these conditions are optimal for all states, and prove that, under this assumption, the value $E_{\mathcal{N},\max}=\log 2$ is maximal for all states and all environments described by the operators~(\ref{lindblad_op}).

\paragraph*{Case of ${c_+ = -c_- \equiv c}$.}
In order for the dissipator to be non-trivial, at least one of the rates $\gamma_k$ must be strictly greater than $0$. Due to the symmetry between the modes, we can choose, with no loss of generality,  $\gamma_1>0$. The equations (\ref{damped_cond}) with $c_+ = -c_- \equiv c$ are thus equivalent to
\begin{equation} \label{damped_abc}
\begin{split}
0= \frac{g_{1}(\vec{z},\vec{t})}{\gamma_1} = & \:
	\left(\chi a - 2 a^2\right) 
	+ \gamma \left(\chi b - 2 b^2\right) 
	+ 2 (1 + \gamma) c^2, \\
0= \frac{g_{2}(\vec{z},\vec{t})}{\gamma_1} = & \: 
	\left[(\gamma a + b) \chi
	- 2(1+\gamma)\left(ab-c^2\right)\right]\left(ab-c^2\right),
\end{split}
\end{equation}
where $\chi \coloneqq (x_0^{-2}+x_0^2)/2 \geqslant 1$ and $\gamma \coloneqq \gamma_2/\gamma_1 \in [0,1]$ (because, again, with no loss of generality we can assume $\gamma_2\leqslant\gamma_1$).

Assuming $a\geqslant b$, it follows from the Heisenberg constraint $h_2(\vec{z})\leqslant 0$ that:
\begin{equation}
\begin{split}
2\leqslant 4(a^2+b^2+2c_-c_+)=4a^2+4b^2-8c^2\leqslant 8a^2,
\end{split}
\end{equation}
and thus $a \geqslant 1/2$. Analogously, if $a<b$, one obtains $b \geqslant 1/2$. Hence, $a,b\geqslant 1/2$ are necessary conditions for the system (\ref{constraints_general_log_neg}, \ref{constraints_general_h}) to be solvable. We therefore can, without loss of generality, use the squeezed state parametrization~(\ref{squeezed_parametrization}).

Once again solving the stabilizability conditions (\ref{damped_abc}), this time for $\nu_\pm$, we obtain
\begin{equation} \label{nu_+-damped}
\begin{split}
\nu_-(r,\chi,\gamma) =&\: \chi
	\frac{\cosh^2 r + \gamma \sinh^2 r}{1+\gamma+(1-\gamma)\cosh 2r}, \\
\nu_+(r,\chi,\gamma) =&\: \chi
	\frac{\gamma \cosh^2 r + \sinh^2 r}{1+\gamma-(1-\gamma)\cosh 2r}.
\end{split}
\end{equation}
Since any two-mode squeezed state fulfills the Heisenberg uncertainty relation (\ref{constraints_general_h}), the system (\ref{constraints_general_log_neg}, \ref{constraints_general_h}) is now reduced to
\begin{equation}
\begin{split}
E_N(r,\chi,\gamma) > 0 
\:\:\textrm{and}\:\: 
1/2 \leqslant \nu_-(r,\chi,\gamma) \leqslant \nu_+(r,\chi,\gamma).
\end{split}
\end{equation}
Solving this system we find that $E_N(r,\chi,\gamma)$ is maximized (only) in the limit $\gamma \to 1$.

\begin{figure*}[!tb]
\begin{subfigure}{0.49\textwidth}
  \includegraphics[width=1\textwidth]{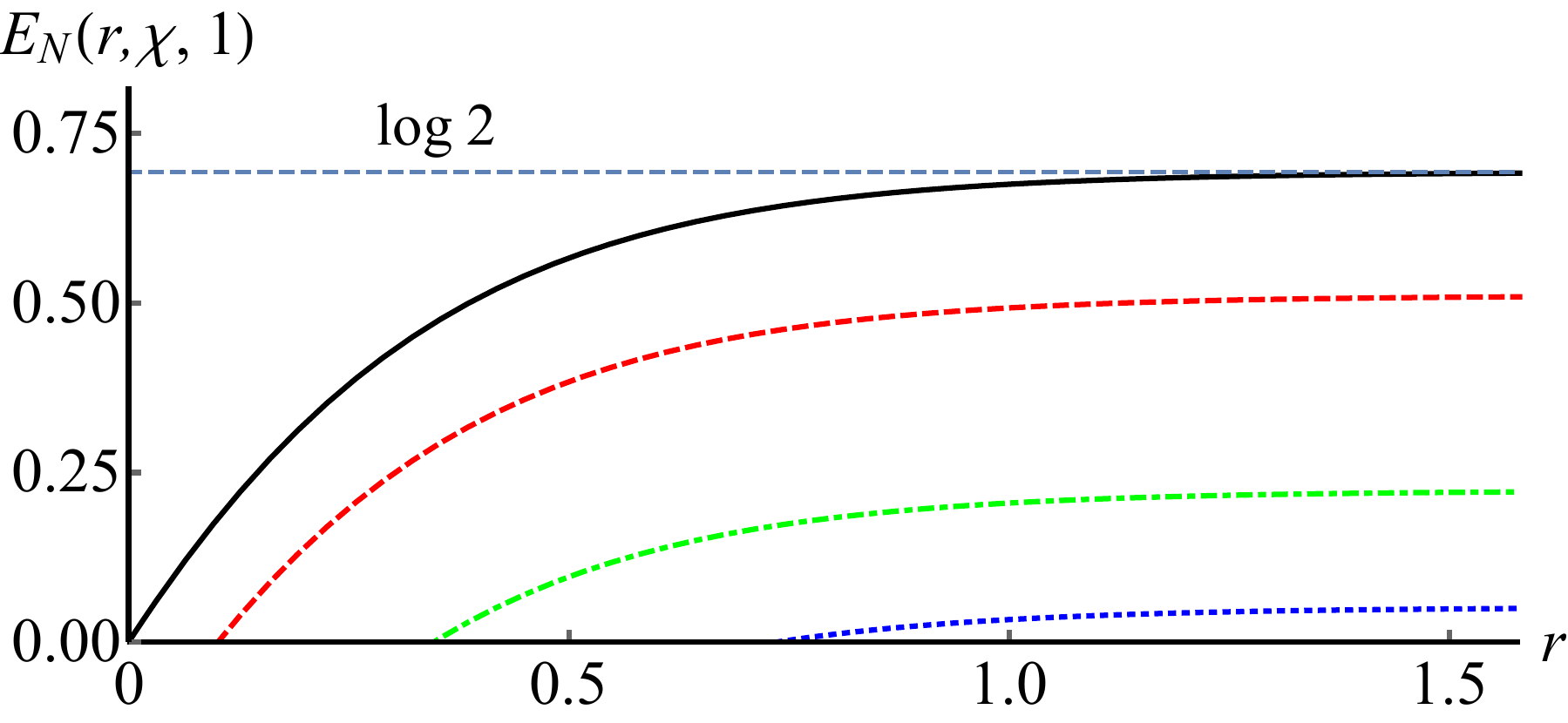}
  \caption{}
  \label{fig:E_N}
\end{subfigure}\hspace{2mm}%
\begin{subfigure}{0.49\textwidth}
  \includegraphics[width=1\textwidth]{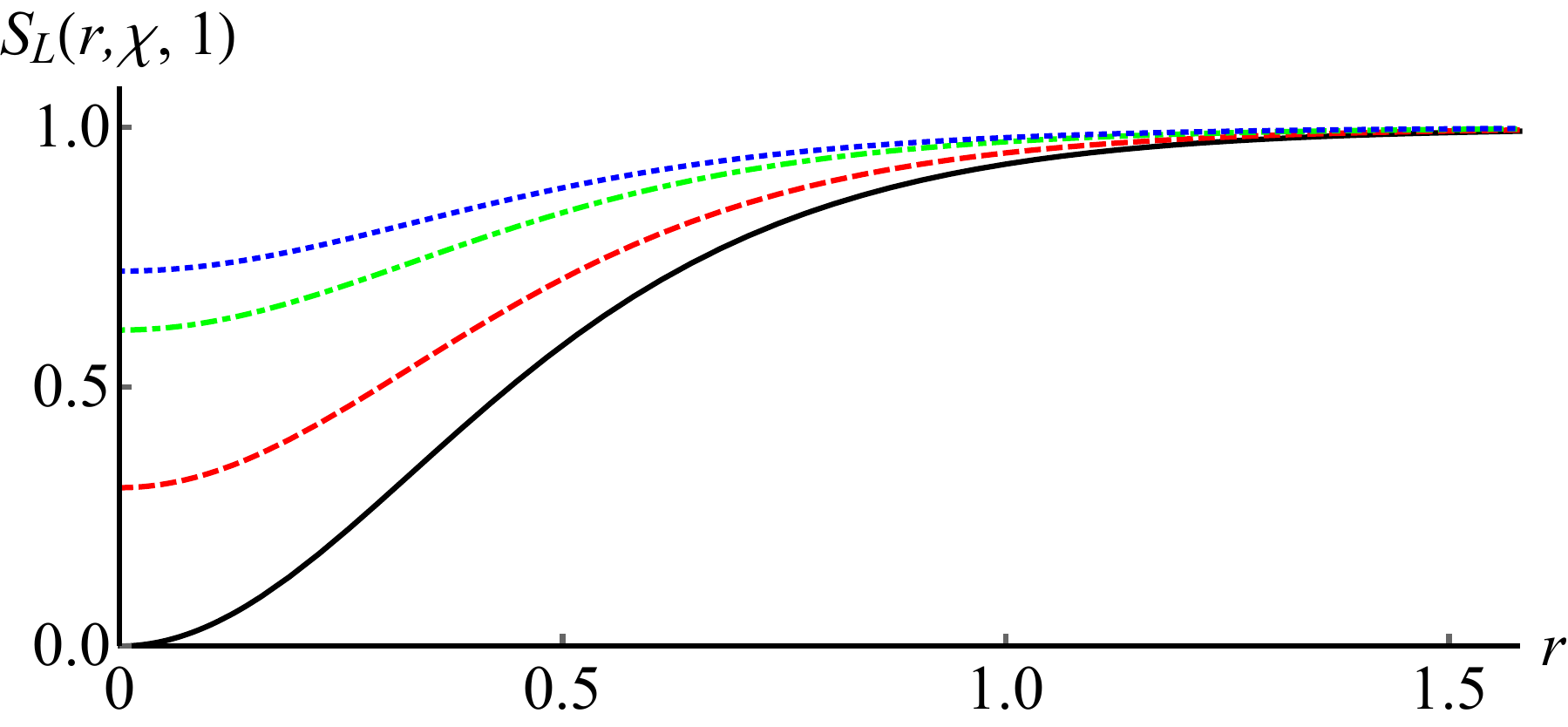}
  \caption{}
  \label{fig:S_L}
\end{subfigure}
\begin{subfigure}{0.49\textwidth}
  \centering
  \includegraphics[width=1\textwidth]{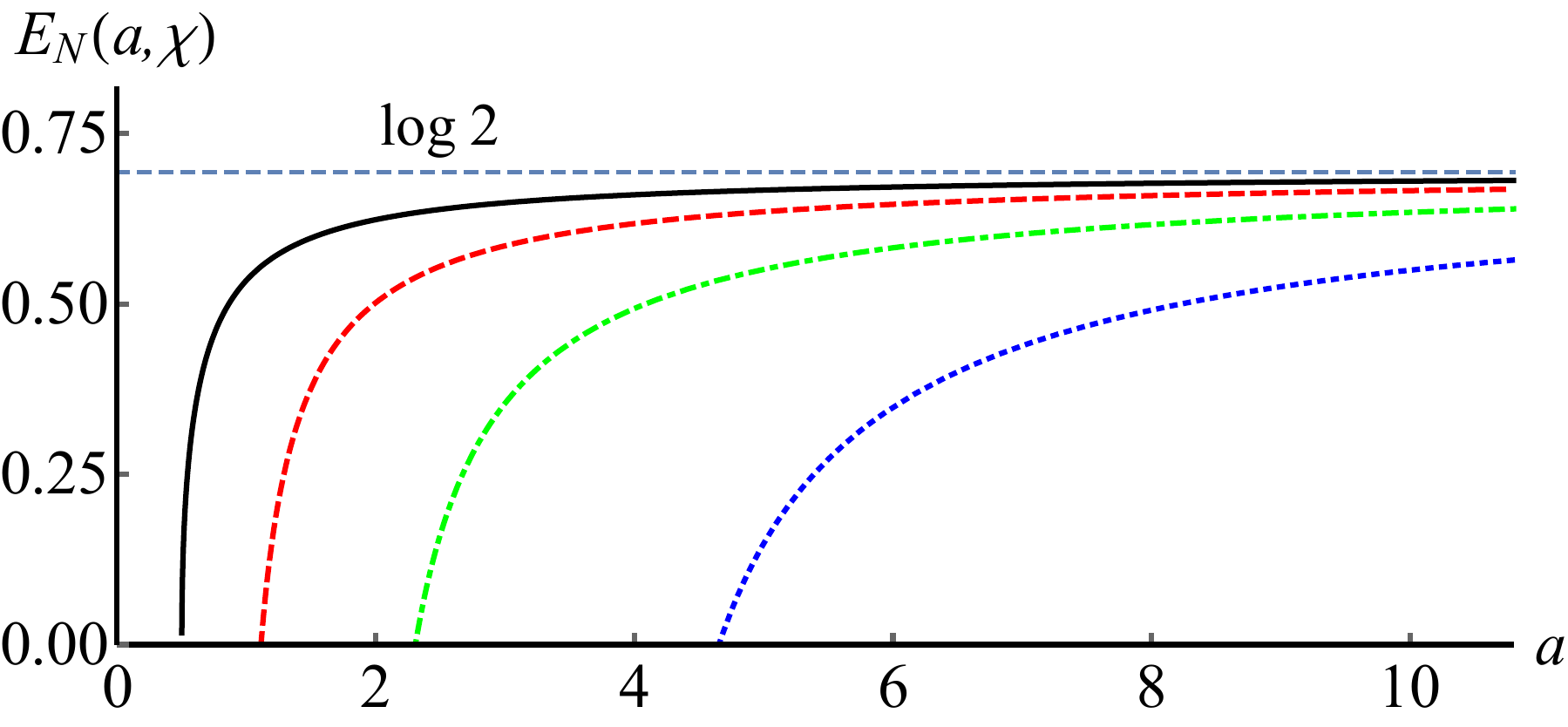}
  \caption{}
  \label{fig:E_N_a=b2}
\end{subfigure}\hspace{2mm}%
\begin{subfigure}{0.49\textwidth}
  \centering
  \includegraphics[width=1\textwidth]{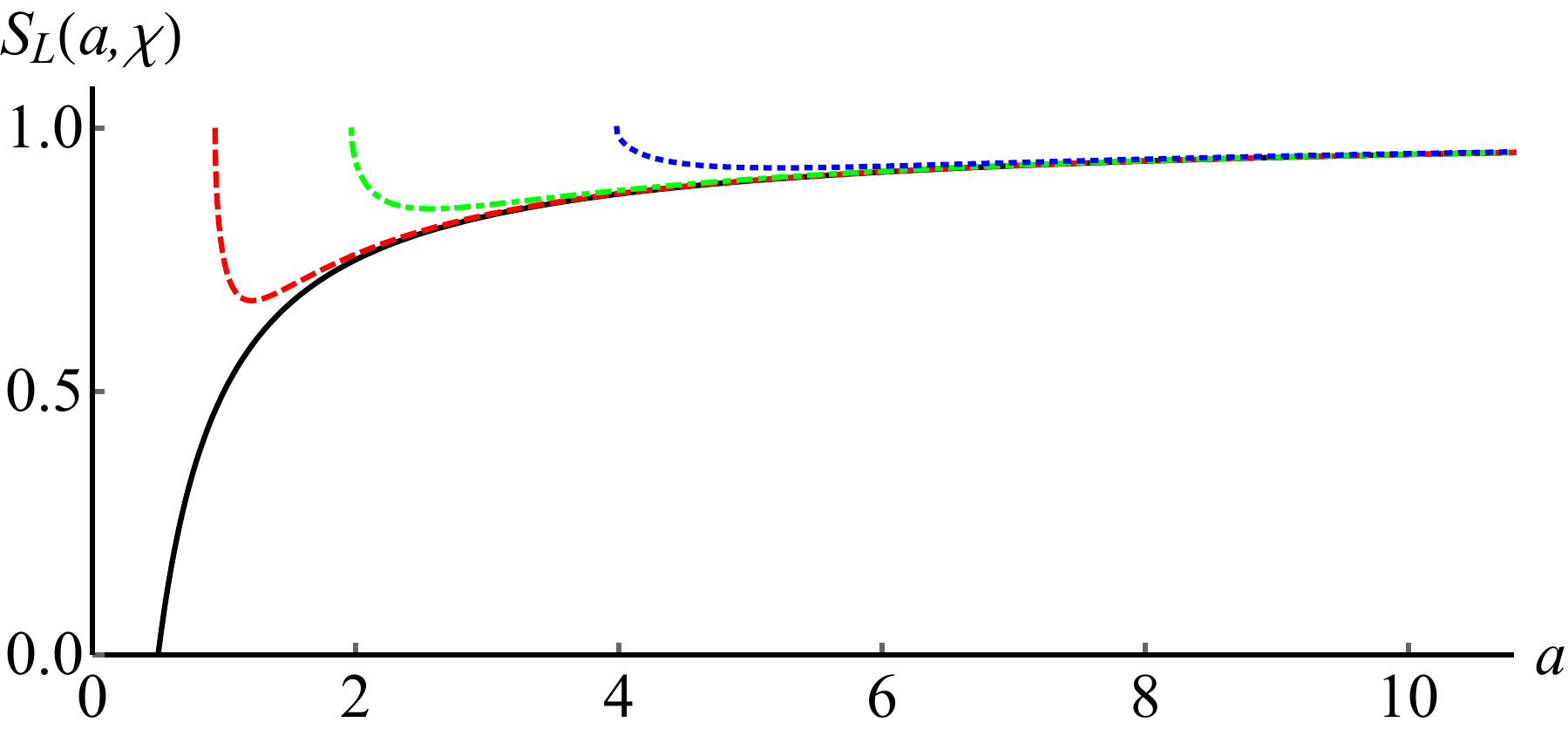}
  \caption{}
  \label{fig:S_L_a=b2}
\end{subfigure}
\captionsetup{width=1\textwidth,justification=centerlast}
\caption{Stabilizable entanglement in the presence of local damping. In the top row, the logarithmic negativity $E_\mathcal{N}(r,\chi,1)$ (a) and the linear entropy $S_L(r,\chi,1)$ (b) are plotted as functions of $r$, for the case $c_+=c_-$, and with four different values of $\chi \in \{1.0,1.2,1.6,1.9\}$ -- solid (black), dashed (red), dot-dashed (green) and dotted line (blue), respectively. In the bottom row, the logarithmic negativity $E_\mathcal{N}(a,\chi)$ (c) and the linear entropy $S_L(a,\chi)$ (d) are plotted as functions of $a$, for the case $a=b$, and with four different values of $\chi \in \{1,2,4,8\}$ -- solid (black), dashed (red), dot-dashed (green) and dotted line (blue), respectively. In both cases we find that, while the dissipator acts adversarial to the entanglement, logarithmic negativities assume positive values, which are bounded from above by $\log 2\approx 0.69$. Moreover, as the logarithmic negativities grow, so do the corresponding linear entropies, indicating an (asymptotic) tradeoff relation between the~entanglement and the purity of stabilizable states.}
\label{fig:E_N_and_S_L}
\end{figure*} 

Using this, we now study the system with $\gamma = 1$. The stabilizable state then becomes symmetric:
\begin{equation} \label{nu_+-damped_simple}
\begin{split}
\nu_\pm(r,\chi,1) = \frac{\chi}{2}\cosh 2r.
\end{split}
\end{equation}
Obviously, this state is always physical, as $\nu_\pm(r,\chi,1)\geqslant 1/2$ for all $r,\chi$. The entanglement condition, on the other hand, leads to the following solution in terms of the squeezing parameter:
\begin{equation}
2r > \artanh \left(\chi-1\right),
\end{equation}
where $\chi\leqslant 2$. As long as this simple criterion is fulfilled, the logarithmic negativity (\ref{log_neg}) is positive and reads
\begin{equation} \label{E_N_damped}
E_\mathcal{N}(r,\chi,1) = \log \big(2/\chi\big) - \log \left(1+e^{-4r}\right).
\end{equation}
Evidently, for a fixed dissipator (fixed value of the characteristic length parameter $\chi$), the maximum is attained in the limit of infinite squeezing 
\begin{equation} \label{limit_damped}
\lim\limits_{r\to\infty} E_\mathcal{N}(r,\chi,1) = \log \big(2/\chi\big),
\end{equation}
which, as we anticipated, is upper bounded by $E_{\mathcal{N},\max}=\log 2$ (for $\chi=1$).

Regarding purity, we find that, despite the symmetry between the two modes: $a=b$, $c_+=-c_-$, the state is highly mixed -- in the sense that its entropy is near-maximal \footnote{Some readers may be familiar with the fact that in finite-dimensional systems all states that are {sufficiently close} to the maximally mixed state are separable \cite{separable_ball}. We stress that this fact does not extend to continuous variable systems, and so there is no inconsistency with our findings connecting the amount of entanglement to the amount of mixedness in stabilizable states.}. Indeed, the linear entropy (\ref{S_L_simple}) takes the form
\begin{equation}
S_L(r,\chi,1)=1-\left(\chi\cosh 2r\right)^{-2}.
\end{equation}
Clearly, $S_L(r,\chi,1)$ rapidly approaches its maximal value $1$ as a function of $r$, regardless of the value of the characteristic length parameter $\chi$. This implies that, independent from the length scale of the system, the only stabilizable entangled states are (highly) mixed. We note that similar results are obtained when considering the Tsallis and R\'{e}nyi entropies.

$E_\mathcal{N}(r,\chi,1)$ and $S_L(r,\chi,1)$ are plotted in Figures \ref{fig:E_N}-\ref{fig:S_L} as functions of $r$ for four different values of $\chi$. We find that the logarithmic negativity assumes a finite, positive value in the limit $r\to\infty$. Notably, regardless of the value of $\chi$, all stabilizable entangled states are characterized by a non-zero degree of mixedness (the only stabilizable pure state is the vacuum state, $r=0$).

\paragraph*{Case of ${a=b}$.} In this case, the stabilizability conditions (\ref{damped_cond}) become effectively independent~of~the rates $\gamma_k$. Solving them for $c_\pm$, as described at the beginning of this section, we obtain two solutions $\big(c_{\pm}\big)_k(a,\chi)$, $k\in\{1,2\}$. The first of these solutions features $c_{+}=-c_{-}$. This is just a special case of the problem solved previously.

The second solution takes the following explicit form:
\begin{equation}
\begin{split}
c_+(a,\chi)=\sqrt{\frac{a (2 a-\chi ) \left[1+2(2 a-\chi)\left(q_\chi+\chi \right)\right]}{8 a \left(q_\chi+\chi \right)-2}},
\end{split}
\end{equation}
where $q_\chi\coloneqq\sqrt{\chi^2-1}$, with the corresponding $c_-(a,\chi)=a(\chi-2a)/\big[2c_+(a,\chi)\big]$. The solution can then be substituted into the system (\ref{constraints_general_log_neg}, \ref{constraints_general_h}), yielding the following constraint:
\begin{equation}
\begin{split}
8a> \left(9 \chi+ 4 \sqrt{3}q_\chi+\sqrt{129 \chi ^2+72\sqrt{3}\chi q_\chi -80}\right).
\end{split}
\end{equation}
The logarithmic negativity and linear entropy read
\begin{equation}
\begin{split}
E_{\mathcal{N}}\left(a,\chi\right)&=
	-\log\sqrt{2a(4a-\chi)-2p_\chi(a)\sqrt{2a(2a-\chi)}},\\
S_{L}\left(a,\chi\right)&=
	1-1/p_\chi(a),
\end{split}
\end{equation}
where $p_\chi(a)\coloneqq 2a\abs{4a-\chi}/\sqrt{16a^2-8\chi a+1}$. Both quantities are monotonically increasing functions of the parameter $a$. As before, the logarithmic negativity is bounded by $E_\mathcal{N,\max}=\log 2$, which is reached in the limit of {extreme} covariance matrices, $a\to\infty$ (this time for all $\chi$). These results are illustrated in Figures \ref{fig:E_N_a=b2}-\ref{fig:S_L_a=b2}, where $E_\mathcal{N}(a,\chi)$ and $S_L(a,\chi)$ are plotted as functions of $a$ for four different values of~$\chi$.

\vspace{5mm} % why???

Our findings for the two cases, $c_+=-c_-$ and $a=b$, suggest that, among all the environments described by Lindblad operators of the form (\ref{lindblad_op}), the preservation of entangled states is {the most efficient} when $\chi=x_0=1$ and $\gamma=\gamma_1/\gamma_2=1$. More precisely, we conjecture that the logarithmic negativity of any state (i.e., for fixed $\vec{z}$) takes its maximum for the dissipator given by $\chi=\gamma=1$.

We now solve the system once again, this time for a general state [no assumptions about $(a,b,c_\pm)$], but for the specific dissipator $\chi=\gamma=1$. We show that the logarithmic negativity is then again bounded from above by $E_{\mathcal{N},\max}=\log 2$. This supports our conjecture that this value is maximal for all states and all environments described by the operators~(\ref{lindblad_op}).

\paragraph*{Case of ${\chi=\gamma=1}$.} Solving (\ref{damped_cond}) for $\big(c_{\pm}\big)_k(a,b)$, $k\in\{1,2\}$, and substituting into the system (\ref{constraints_general_log_neg}, \ref{constraints_general_h}), we obtain the solution $a\geqslant 1/2$, $a=b$. This is a special case of the problem solved above. Thus, we conclude that, under the assumption that, for a given state, the logarithmic negativity is maximal when $\chi=\gamma=1$, the value $E_{\mathcal{N},\max}=\log 2$ is maximal for all states subject to dissipators described by the operators (\ref{lindblad_op}).

An example Hamiltonian, that stabilizes states characterized by $E_{\mathcal{N}}=E_{\mathcal{N},\max}=\log 2$, is given by 
\begin{equation} \label{eq:hamiltonian_squeezed}
\begin{split}
\hat{H}_{\textnormal{sq}}=
    -i\hbar\omega\big(\hat{a}_1\hat{a}_2-\hat{a}_1^\dag\hat{a}_2^\dag\big),
\end{split}
\end{equation}
where $\omega$ is a positive constant defining the energy levels of the system. The resulting unitary evolution is governed by the squeezing operator (\ref{squeezing_operator}). In other words, our analysis shows that, in the model of local damping, no other quadratic Hamiltonian can outperform the squeezing Hamiltonian (\ref{eq:hamiltonian_squeezed}) in stabilizing entanglement.

\subsection{Dissipative squeezed-state preparation}
\label{sec:two_mode_squeezed_state_preserving_dissipation}
We now discuss dissipators, which are designed to produce two-mode squeezed states, arising from applying the squeezing operator (\ref{squeezing_operator}) with $r=\alpha$ to the two-mode vaccuum state \cite{using_dissipation_2}. In other words, these dissipators are specifically engineered to preserve two-mode squeezed states with $r=\alpha$. Such models have been discussed in the context of the experimental generation of entanglement \cite{squeezed_dissipation_experimental_1,squeezed_dissipation_experimental_2}. 

\begin{figure*}[!tb]
  \centering
\begin{subfigure}{0.49\textwidth}
  \centering
  \includegraphics[width=1\textwidth]{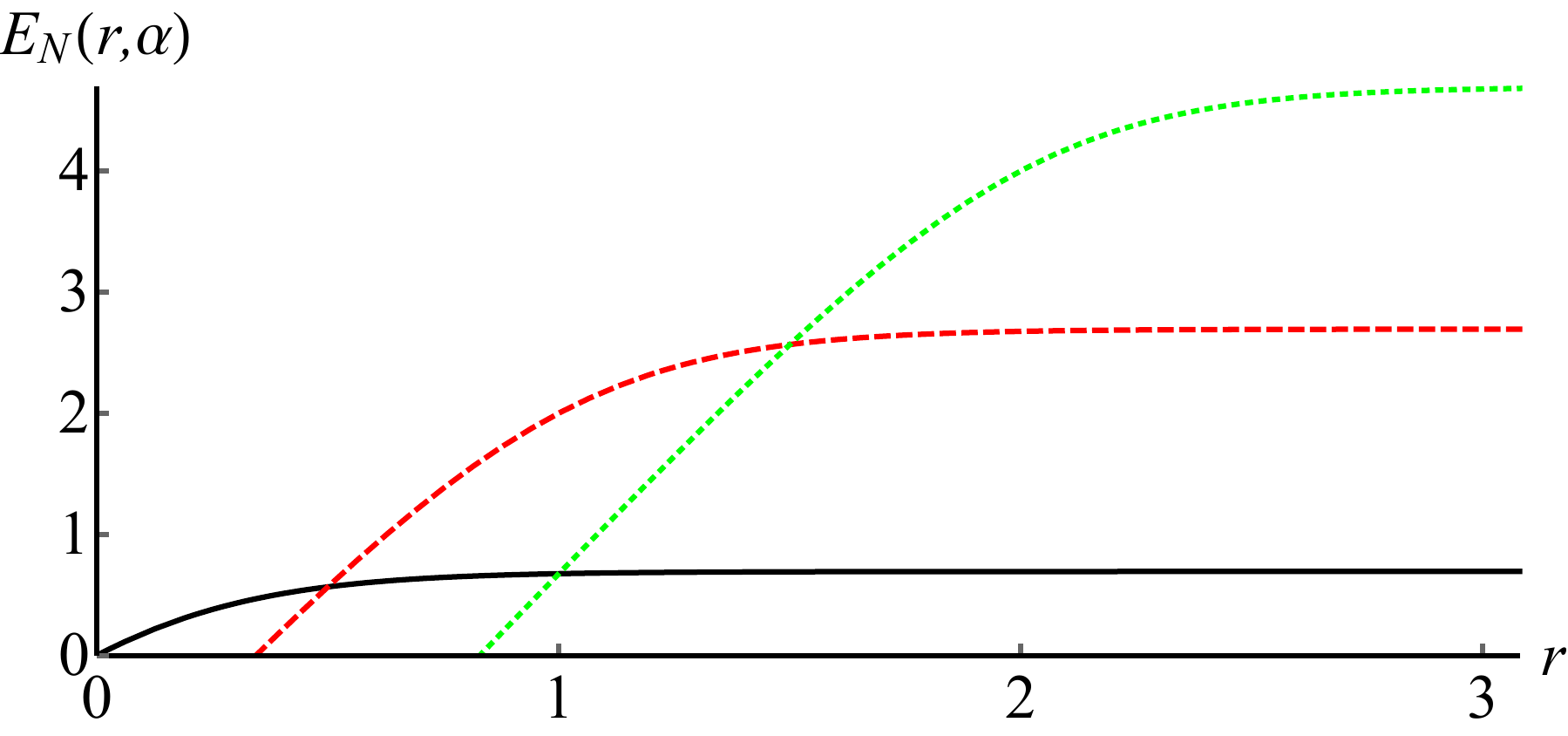}
  \caption{}
  \label{fig:E_N_two-mode}
\end{subfigure}\hspace{2mm}%
\begin{subfigure}{0.49\textwidth}
  \centering
  \includegraphics[width=1\textwidth]{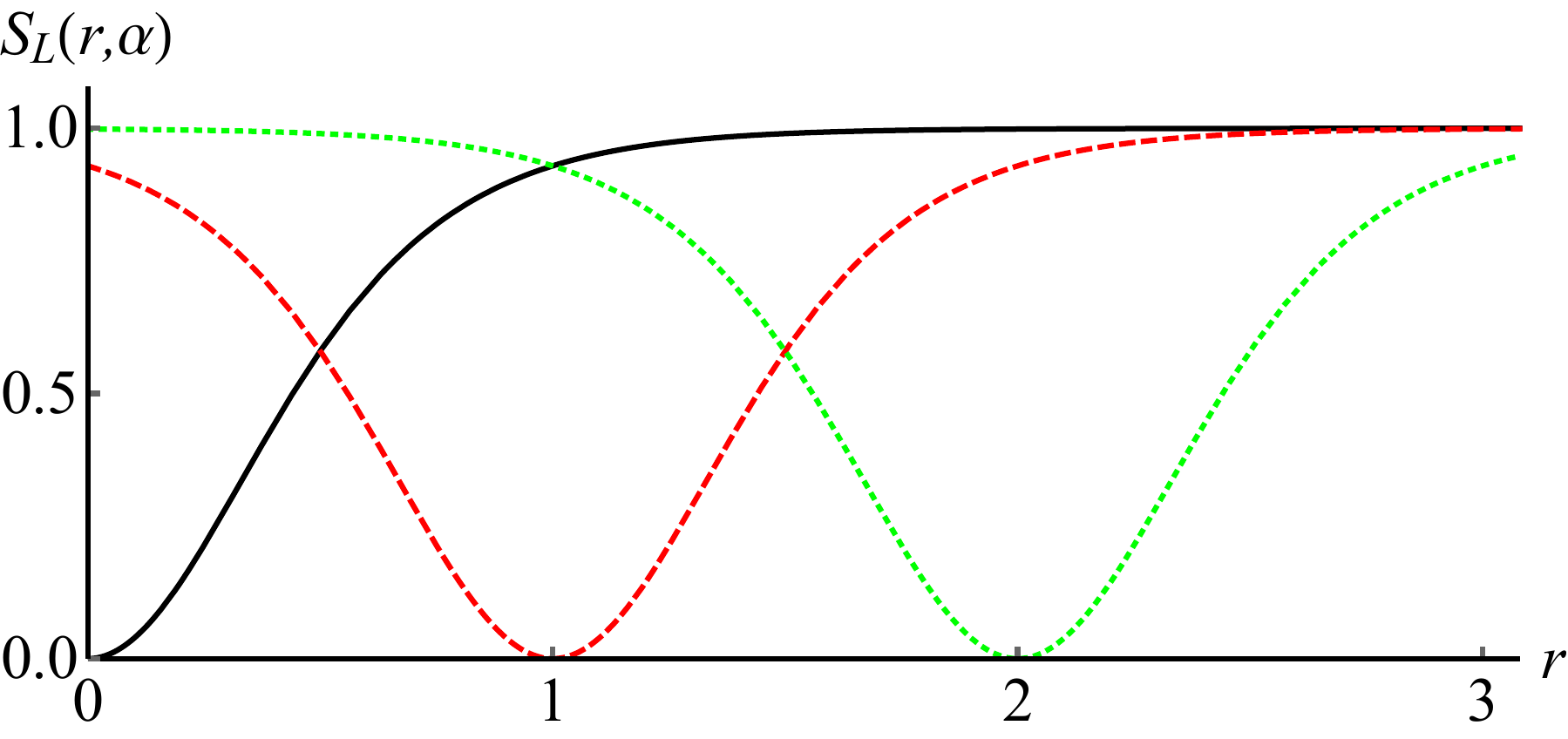}
  \caption{}
  \label{fig:S_L_two-mode}
\end{subfigure}
\captionsetup{width=1\textwidth,justification=centerlast}
\caption{Stabilizabile entanglement for dissipators engineered to preserve two-mode squeezed thermal states. Shown are the logarithmic negativity $E_\mathcal{N}(r,\alpha)$ (a) and the linear entropy $S_L(r,\alpha)$ (b) as functions of $r$, for three different values of $\alpha \in \{0,1,2\}$ -- solid (black), dashed (red) and dotted (green), respectively. We find that, irrespective of the nonlocal character of the dissipator, the amount of stabilizable entanglement is finite and bounded from above by $2\alpha+\log 2$, a value $\log 2$ greater than the amount of entanglement in the dissipator's {dedicated} two-mode squeezed state. The states achieving this optimal value are close to maximally mixed, while the linear entropies assume their minima at their respective {dedicated} two-mode squeezed states.}
\label{fig:E_N_and_S_L_two-mode}
\end{figure*}

By construction, the dissipator stabilizes the two-mode squeezed state characterized by $r=\alpha$. Consequently, the latter describes the steady state of the system in the absence of a Hamiltonian. However, the model also admits other stabilizable states, possibly characterized by higher entanglement. In this section, we demonstrate that this is indeed the case.

In principle, one could consider only single-mode squeezing, see, e.g., \cite{using_dissipation_2}. Here, we focus on full two-mode squeezing, induced by the two Lindblad operators
\begin{equation} \label{dis_two-modes}
\begin{split}
\hat{L}_1 & \coloneqq 
    \cosh\alpha\, \hat{a}_1 - \sinh\alpha\, \hat{a}_2^\dag,\\
\hat{L}_2 & \coloneqq
	\cosh\alpha\, \hat{a}_2 - \sinh\alpha\, \hat{a}_1^\dag,
\end{split}
\end{equation}
where $\alpha \geqslant 0$. The resulting dissipator consists of two {channels}, each creating a superposition of states in which one of the modes gains and the other loses a particle, with the rate of the losses and gains controlled by the parameter $\alpha$.

{As mentioned above, our objective is to show that there exist stabilizable states that are more entangled than the dissipator's dedicated squeezed thermal states. To this end, it is sufficient to consider the special case $c_+ = -c_- \equiv c$, which includes the aforementioned squeezed states.} As discussed in the previous subsection, we can then use the squeezed thermal state parametrization (\ref{squeezed_parametrization}) with no loss of generality. The conditions (\ref{stabilizability}) assume the form
\begin{equation}
\begin{split}
0&=g_{1}(\nu_\pm,r,\alpha) = 2 \left(\nu _-^2+\nu _+^2\right)
	-\left(\nu _-+\nu _+\right) \cosh 2 (r-\alpha), \\
0&=g_{2}(\nu_\pm,r,\alpha) = \nu _- \nu _+ 
	\left[4 \nu _- \nu _+-\left(\nu _-+\nu _+\right) \cosh 2(r-\alpha)\right],
\end{split}
\end{equation}
where, just as in the case of two modes with local damping, we simplified $g_{2}(\vec{z})$ using $g_{1}(\vec{z})=0$. Comparing the $\cosh 2 (r-\alpha)$ terms in the two equations, one can easily see that they can be simultaneously fulfilled if and only if $\nu_-=\nu_+\equiv \nu$. This immediately leads to the solution:
\begin{equation}
\begin{split}
\nu(r,\alpha) = \frac{1}{2}\cosh 2(r-\alpha).
\end{split}
\end{equation}
The corresponding logarithmic negativity (\ref{log_neg}) is equal to
\begin{equation} \label{E_N_two-mode}
\begin{split}
E_\mathcal{N}(r,\alpha)\coloneqq 
	- \log \left[e^{-2r}\cosh 2(r-\alpha)\right].
\end{split}
\end{equation}
Clearly, the state is always physical, as $\nu(r,\alpha)\geqslant 1/2$ for all $r,\alpha$. As for the presence of entanglement, it follows from the definition (\ref{log_neg}) that the state is entangled if and only if the argument of the above logarithm is smaller than 1. This leads to the following condition: 
\begin{equation}
\begin{split}
4r >  2\alpha- \log(2-e^{-2\alpha}).
\end{split}
\end{equation}
For a fixed dissipator (fixed $\alpha$), we have
\begin{equation} \label{eq:log_neg_sq}
\begin{split}
0 \leqslant E_\mathcal{N}\left(r,\alpha\right)
 	\leqslant \lim\limits_{r\to\infty} E_\mathcal{N}\left(r,\alpha\right) =
	\log 2 + 2\alpha,
\end{split}
\end{equation}
obtainable, e.g., with a Hamiltonian of~the~form~(\ref{eq:hamiltonian_squeezed}). 

We make the following observations: firstly, it is clear that, despite the nonlocal character of the Lindblad operators (\ref{dis_two-modes}), arbitrarily high entanglement can only be obtained in the limit $\alpha\to\infty$. Secondly, and perhaps more interestingly, the value (\ref{eq:log_neg_sq}) is $\log 2$ higher than the logarithmic negativity of the two-mode squeezed state with $r=\alpha$, which the dissipator is engineered to produce by default. In other words, there exist states stabilizable with respect to the dissipator that are more entangled than the {dedicated} two-mode squeezed state. Finally, we can see that, for $\alpha=0$, the maximum negativity is equal to $E_{\mathcal{N},\max}=\log 2$. In fact, one can easily check that, when $\alpha=0$, the logarithmic negativity (\ref{E_N_two-mode}) is exactly equal to that in eq. (\ref{E_N_damped}) with $\chi=1$. This is what we should expect based on the discussion in the previous subsection, as in this case the operators (\ref{dis_two-modes}) coincide with those in (\ref{lindblad_op}) with $x_0=\gamma_2/\gamma_1=1$. Similar results hold for the entropies, in particular the linear entropy
\begin{equation}
\begin{split}
S_L(r,\alpha)\coloneqq \tanh^2 2(r-\alpha).
\end{split}
\end{equation}

The logarithmic negativity and the linear entropy are both plotted in Figure \ref{fig:E_N_and_S_L_two-mode} as functions of $r$ for four different values of $\alpha$. As in the previous models, the logarithmic negativity rapidly approaches its maximum value, $\log 2 + 2\alpha$. We stress again that this maximum value is $\log 2 \approx 0.69$ higher than the logarithmic negativity of the two-mode squeezed state with $r=\alpha$.

The behaviour of the linear entropy deviates from the previous models. We find that, in the neighbourhood of the point $r=\alpha$, there exist highly entangled states that are (nearly) pure. This is simply a consequence of the fact that the dissipator (\ref{dis_two-modes}) is designed to preserve pure two-mode squeezed states with $r=\alpha$. Irrespectively, we find that for a fixed environment (fixed $\alpha$), stabilizable states which maximize entanglement are close to maximally mixed.

\paragraph*{Local perturbation.} We complement our analysis by considering local perturbations of the dissipator (\ref{dis_two-modes}). As argued above, the presence of some local dissipation is usually unavoidable in realistic scenarios. Depending on the strength of the local noise, we must expect that our results regarding the stabilizability of entangled states are adjusted.

To account for this fact, we modify our model by adding two Lindblad operators for local damping (\ref{lindblad_op}), with $\gamma_1=\gamma_2\equiv\eta$ responsible for the relative strength of the local dissipation, and $x_0=1$ for simplicity. The resulting logarithmic negativity reads
\begin{equation}
E_\mathcal{N}(r,\alpha,\eta) = E_\mathcal{N}(r,\alpha)
    -\log\frac{1+\eta\cosh 2r\cosh^{-1}2(r-\alpha)}{1+\eta},
\end{equation}
where $E_\mathcal{N}(r,\alpha)$ refers to the logarithmic negativity of the unmodified model (\ref{E_N_two-mode}). Clearly, regardless of the parameter $\alpha$ of the dissipator, for a fixed state (fixed $r$), the logarithmic negativity is lowered by the presence of local noise. This is in line with our intuition that local dissipation should reduce the stabilizable entanglement.

\subsection{Cascaded oscillators}\label{sec:cascaded_oscillators}
We finally discuss the case of {cascaded oscillators} coupled to the vacuum \cite{using_dissipation_2}. The use of cascaded oscillators is common in experimental setups, ranging from the production of entangled states \cite{OPOs_entanglement} to spectroscopy \cite{OPOs_spectroscopy}. The particular model under consideration has recently been discussed in the context of entanglement distribution \cite{OPOs_computation}. Moreover, this form of mode coupling is leveraged in the Coherent Ising Machine \cite{ising_machine_1}.

\begin{figure*}[!tb]
  \centering
\begin{subfigure}{0.49\textwidth}
  \centering
  \includegraphics[width=1\textwidth]{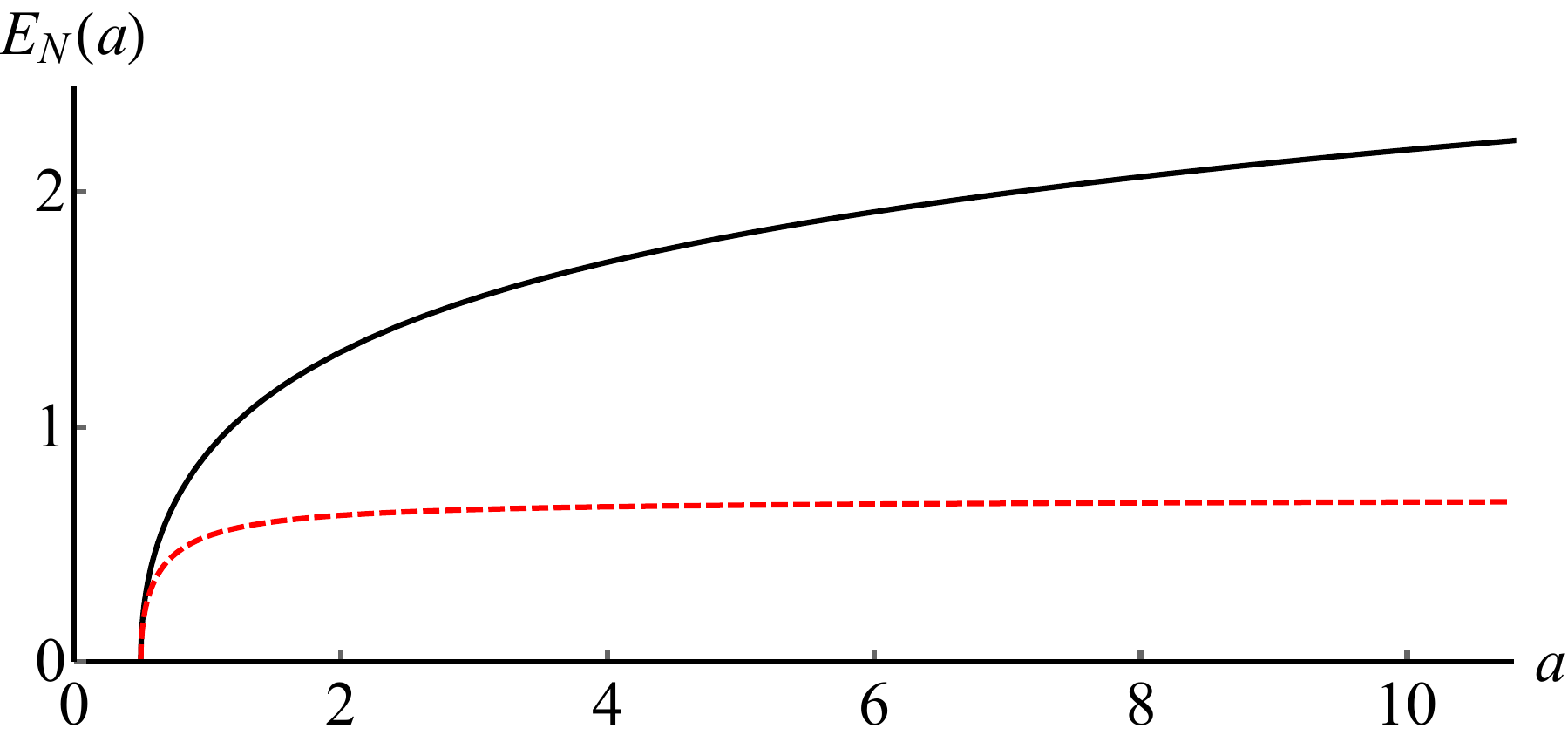}
  \caption{}
  \label{fig:E_N_cascaded}
\end{subfigure}\hspace{2mm}%
\begin{subfigure}{0.49\textwidth}
  \centering
  \includegraphics[width=1\textwidth]{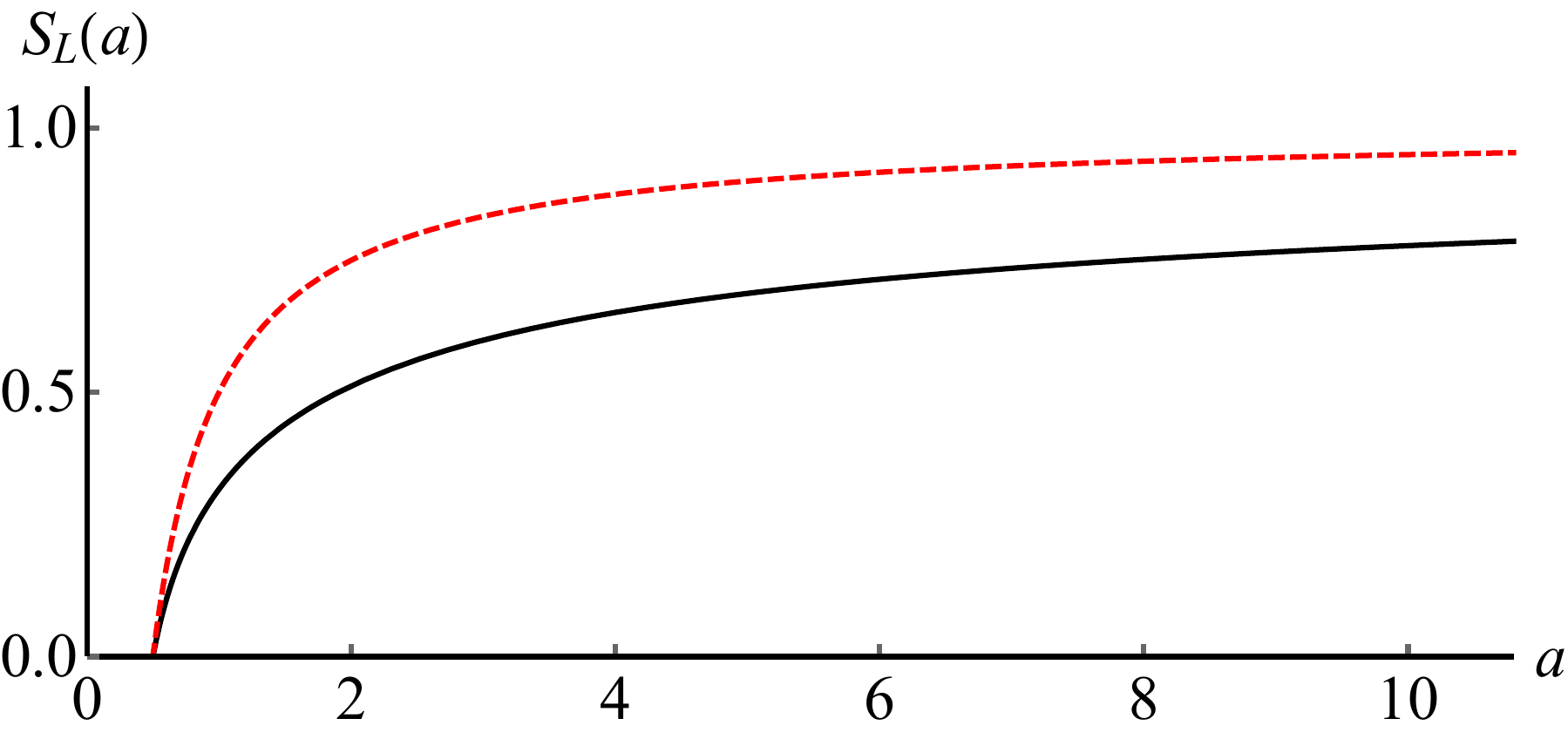}
  \caption{}
  \label{fig:S_L_cascaded}
\end{subfigure}
\captionsetup{width=1\textwidth,justification=centerlast}
\caption{Stabilizable entanglement for two cascaded oscillators coupled to the vacuum. The logarithmic negativity $E_{\mathcal{N}}(a,c_+)$ (a) and the linear entropy $S_{L}(a,c_+)$ (b) are plotted as functions of $a$, with $c_+=c_{+,\textrm{max}}(a)$ and $c_+=c_{+,\textrm{mid}}(a)$ -- solid (black) and dashed line (red), respectively. Reflecting the global character of the dissipator, we find that the amount of entanglement for the states characterized by $c_+(a)=c_{+,\textrm{max}}(a)$ grows unbounded as~$a\to\infty$. States characterized by maximum logarithmic negativity are close to the maximally mixed state.}
\label{fig:E_N_and_S_L_cascaded}
\end{figure*} 

In this scheme, we have a single Lindblad operator
\begin{equation} \label{dis_cascaded}
\begin{split}
\hat{L} & \coloneqq \sqrt{\kappa}\left(\hat{a}_1+\hat{a}_2\right),
\end{split}
\end{equation}
where $\kappa>0$ is a parameter responsible for the strength of the dissipation. The model is similar to the one discussed in Section \ref{sec:two_modes_with_local_damping} with $\gamma=x_0=1$, in the sense that the interaction with the environment results in the loss of excitations in the modes. However, while that dissipator consisted of two {channels}, each of which decreased the number of excitations in one of the modes in a deterministic fashion, here, the dissipator consists of only one {channel}, whose action on the state creates a superposition of two states, each with an excitation lost in one of the modes. Interestingly, the steady state of the system in the absence of a Hamiltonian, given by $b=a$, $c_\pm=1/2-a$, is {separable}.

{In the model at hand, our main objective is to demonstrate that, unlike in the previous models, it is possible to achieve infinite logarithmic negativity. To this end, it is sufficient to consider states characterized by $a=b$. The definition (\ref{dis_cascaded}) then gives rise to the following stabilizability conditions (\ref{stabilizability}):}
\begin{equation} \label{eq:cond_2}
\begin{split}
0=\frac{g_{1}(\vec{z})}{\kappa} = & \:
	4 a^2 + 2a\left(2 c_+ + 2c_- -1\right)
	+ 4 c_+ c_- - c_+ - c_-,\\
0=\frac{g_{2}(\vec{z})}{\kappa} = & \: -(a+c_+)(a+c_-)\frac{g_{1}(\vec{z})}{\kappa}.
\end{split}
\end{equation}
Note that the value of $\kappa$ is irrelevant for stabilizability. This is an immediate consequence of the fact that the conditions (\ref{stabilizability}) are linear in $C^\dag C$. Physically, it corresponds to the fact that the overall dissipation strength merely affects the transition time to the steady state, not the steady state itself.

Clearly, the two equalities (\ref{eq:cond_2}) are valid only if $g_{1}(\vec{z})=0$. Solving for $c_-$, we obtain
\begin{equation}
\begin{split}
c_-(a,c_+)=&\:-a+(a+c_+)/(4 a+4 c_+-1).
\end{split}
\end{equation}
The system (\ref{constraints_general_log_neg}, \ref{constraints_general_h}) is then solved if and only if $a \geqslant 1/2$ and
\begin{equation} \label{c_+(a)}
\begin{split}
c_{+,\min}(a)<c_+\leqslant c_{+,\max}(a),
\end{split}
\end{equation}
where
\begin{equation}
\begin{split}
c_{+,\min}(a) \coloneqq &\: \sqrt{(a-1) a+1/2}-1/2, \\
c_{+,\max}(a) \coloneqq &\: 
	\frac{a-1/2+ \sqrt{2a (2 a-1) (4 a-1) (4 a+1)}}{8 a-1}.
\end{split}
\end{equation}
The logarithmic negativity takes the form
\begin{equation}
\begin{split}
E_\mathcal{N}\left(a,c_{+}\right) \coloneqq -\log
	\left(2\sqrt{\frac{a^2-c_+^2}{4 a+4 c_+-1}}\right).
\end{split}
\end{equation}
Notably, it is a monotonically non-decreasing function of $a$, with the rate of growth proportional to how close $c_+$ is to $c_{+,\max}(a)$. In particular, $E_\mathcal{N}\left[a,c_{+,\min}(a)\right] = 0$ and
\begin{equation}
\begin{split}
\lim\limits_{a\to\infty}E_\mathcal{N}\left[a,c_{+,\max}(a)\right] = \infty.
\end{split}
\end{equation}
In other words, in the limit $a\to\infty$, it is possible to stabilize states characterized by arbitrarily high entanglement. An exemplary Hamiltonian stabilizing such states is given by
\begin{equation}
\begin{split}
\hat{H}_{\textnormal{cas}}=
    (-i\hbar\omega/2)\left[\big(\hat{a}_1+\hat{a}_2\big)^2-
    \big(\hat{a}_1^\dag+\hat{a}_2^\dag\big)^2\right],
\end{split}
\end{equation} 
where $\omega$ is an arbitrary positive constant. One easily checks that, in the limit $a\to\infty$, the functions $c_{+,\max}(a)$, $c_{-,\max}(a)$ practically coincide. By virtue of eq. (\ref{squeezed_parametrization}), we can thus interpret this limit as infinite two-mode squeezing.

For the sake of completeness, we also analyze the case of $c_+ = -c_- \equiv c$, as in the case of local damping. As it turns out, such an assumption leads to $a=b$, effectively reducing it to a special case of above model, with $c_{+,\textrm{mid}}(a)\coloneqq\sqrt{a(a-1/2)}$. In~the limit of infinite squeezing this yields
\begin{equation}
\begin{split}
\lim\limits_{a\to\infty}E_\mathcal{N}\left[a,c_{+,\textrm{mid}}(a)\right] = \log{2},
\end{split}
\end{equation}
a reiteration of the result (\ref{limit_damped}) for local damping.

In all cases, the state is at least partially mixed. The linear entropy (\ref{S_L_simple}) is equal to
\begin{equation}
\begin{split}
S_L(a,c_+)=1-\frac{4 a+4 c_+-1}
	{4\left(a+c_+\right)\sqrt{\left(a-c_+\right)o(a,c_+)}},
\end{split}
\end{equation}
where $o(a,c_+)\coloneqq c_+(8 a-1)+a(8 a-3)$. As is evident from (\ref{c_+(a)}), $c_+$ is at least linear in $a$. The negative term thus eventually decays to $0$ as $a$ grows. The Tsallis and R\'{e}nyi entropies yield similar results.

Figure \ref{fig:E_N_cascaded} shows a comparison of the logarithmic negativities $E_\mathcal{N}(a,c_+)$ with $c_{+,\textrm{max}}(a)$ and $c_{+,\textrm{mid}}(a)$ as input. We find that the former grows indefinitely, while the latter rapidly reaches its maximal value, $E_\mathcal{N,\max}=\log{2}$. In Figure \ref{fig:S_L_cascaded}, we provide an analogous comparison for the corresponding entropies.

The presence of additional, local noise can be taken into account in a similar way as in the case discussed in the previous subsection.

\section{Concluding remarks}
\label{sec:Summary}
We studied the stabilizability of entangled two-mode Gaussian states in three physically motivated dissipative scenarios. Based on a Hamiltonian-independent treatment, we find explicit parametrizations of the stabilizable states in all three models, allowing us to quantify their entanglement and mixedness.

In the case of two modes with local damping, where the dissipator acts adversarial to entanglement, we provide strong evidence that the logarithmic negativity does not exceed $\log 2$ for all stabilizable states. Perhaps counterintuitively, we obtain a similar result in the case of nonlocal dissipators engineered to preserve squeezed thermal states, where an analogous upper bound is derived. For this class of dissipators, we also showed that there exist stabilizable states with entanglement higher than in the case of {dedicated} two-mode squeezed states. In the case of cascaded oscillators coupled to the vacuum, we find that arbitrarily high entanglement can be stabilized. Generally, we observe that, regardless of the model at hand, the stabilizable states which maximize entanglement are close to maximally mixed, indicating an asymptotic tradeoff relation between entanglement and purity among stabilizable states. 

Our findings suggest the following directions for future research. Firstly, we focused here on two-mode Gaussian states. It would be interesting to see how the analysis can be extended to other types of systems; for instance, $N$-mode Gaussian states or non-Gaussian states. {In the former case, Theorem \ref{theorem_covariance} significantly reduces the number of stabilizability conditions. However, due to the lack of a standard form for $N>2$, covariance matrices depend on large numbers of parameters, rendering analytical treatments challenging. In the latter case, a challenge may lie in finding viable parametrizations for families of potentially stabilizable states. Irrespectively, let us point out that, in a more precise sense, our work addresses the stabilizability of covariance matrices, which can also be attributed to non-Gaussian states. For Gaussian states, the covariance matrices comprise the complete state information, including their entanglement properties. When applied to non-Gaussian states, our results still hold in a similar way, with some conclusions weakened (e.g., entanglement criteria based on the covariance matrix are only necessary for non-Gaussian states \cite{Vogel05}).}

Secondly, our conjecture regarding the (absence of) purity of {maximally} entangled stabilizable states relies on specific models of environment. Is it possible to make this statement more rigorous, e.g. by proving it for arbitrary dissipators/systems? Moreover, the theory of stablizability itself may be developed further. For example, the known conditions for stabilizability \cite{stabilizability_geometric,stabilizability_cv_systems} are necessary but not sufficient for all quantum states. Necessary and sufficient conditions, on the other hand, would allow us to draw more stringent conclusions. {Finally, there exist practical scenarios in which the dissipator is at least partially controlled. It may be beneficial to generalize the theory of stabilizability to such scenarios, e.g. by splitting the dissipator into a controllable and non-controllable part, and formulating new conditions for stabilizability with respect to the latter.}

\paragraph*{Acknowledgements.}
Tomasz Linowski and {\L}ukasz Rudnicki would like to acknowledge support by the Foundation for Polish Science (IRAP project, ICTQT, contract no. 2018/MAB/5, co-financed by EU within Smart Growth Operational Programme).

\appendix
\section{Discussion of the general problem of local damping}
\label{sec:explicit_solution_to_the_general_damping_problems}
Solving the stabilizability conditions (\ref{damped_cond}) as described at the beginning of Section \ref{sec:two_modes_with_local_damping}, we recover the solution~as
\begin{equation} \label{eq:cpcm_explicit_damping}
\begin{split}
    c_{+}&=\frac{\gamma _1 \left(-4 a^2 x_0^2+x_0^4 (a+b)+a\right)+b \gamma _2 \left(1-4 b x_0^2\right)}{4 c_- \left(\gamma _1+\gamma _2\right) x_0^2},\\
    (c_-^2)_\pm&=\frac{-B\pm\sqrt{B^2-4AC}}{2A},
\end{split}
\end{equation}
where
\begin{equation*}
\begin{split}
    A=\frac{1}{2} \gamma _1 \left(x_0^2 (a+b)-4 a b\right)-2 a b \gamma _2,
\end{split}
\end{equation*}
\begin{equation*}
\begin{split}
    B=\frac{\gamma _1^2 D+2\gamma _2 \gamma _1 b E+\gamma _2^2 b F}{8 \left(\gamma _1+\gamma _2\right) x_0^4},
\end{split}
\end{equation*}
\begin{equation*}
\begin{split}
    C=\frac{G^2H}{32 \left(\gamma _1+\gamma _2\right){}^2 x_0^6},
\end{split}
\end{equation*}
\begin{equation*}
\begin{split}
    D=\:2 a x_0^4 \left(8 a^3+8 a b^2+a+b\right)-4 a x_0^2 \left(2 a^2+b^2\right)\\
        +a^2+x_0^8 (a+b)^2-4 a x_0^6 (a+b) (2 a+b),
\end{split}
\end{equation*}
\begin{equation*}
\begin{split}
    E=\:x_0^4 \left(32 a^2 b+a+b\right)-2 x_0^6 (a+b) (a+2 b)\\
        -6 a x_0^2 (a+b)+a,
\end{split}
\end{equation*}
\begin{equation*}
\begin{split}
    F=\:\left(4 b x_0^2-1\right) \left[4 x_0^2 \left(a^2+b^2\right)-b\right],
\end{split}
\end{equation*}
\begin{equation*}
\begin{split}
    G=\:\gamma _1 \left(-4 a^2 x_0^2+x_0^4 (a+b)+a\right)+b \gamma _2 \left(1-4 b x_0^2\right),
\end{split}
\end{equation*}
\begin{equation*}
\begin{split}
    H=\:b \gamma _1 \left(4 a x_0^2-1\right)+a \gamma _2 \left(4 b x_0^2-1\right).
\end{split}
\end{equation*}
The constraints (\ref{constraints_general_log_neg}, \ref{constraints_general_h}), on the other hand, read explictly:
\begin{equation*}
\begin{split}
    0&> -4 \left(a^2+b^2-2 c_- c_+\right)+16 \left(a b-c_-^2\right) \left(a b-c_+^2\right)+1,\\
    0&\geqslant 4 \left(a^2+b^2+2 c_- c_+\right)-16 \left(a b-c_-^2\right) \left(a b-c_+^2\right)-1,\\
    0&\geqslant 2-4 \left(a^2+b^2+2 c_- c_+\right).
\end{split}
\end{equation*}
Using the formulas (\ref{eq:cpcm_explicit_damping}), we obtain a set of three inequalities for two independent variables $a$, and $b$, and three parameters $\gamma_1$, $\gamma_2$, and $x_0$ (two if $\gamma = \gamma_2/\gamma_1$ is used). Unless specific values are assigned to these parameters, the solution has to be attempted numerically.
\vspace{0cm}

\bibliography{stabilizability}{}

\begin{thebibliography}{10}

\bibitem{time_as_entanglement}
E.~Moreva, G.~Brida, M.~Gramegna, V.~Giovannetti, L.~Maccone and M.~Genovese,
  \emph{Time from quantum entanglement: An experimental illustration}, Phys.
  Rev. A, \textbf{89}, 052122 (2014).

\bibitem{superdense_coding}
C.~H. Bennett and S.~J. Wiesner, \emph{Communication via one- and two-particle
  operators on {Einstein}-{Podolsky}-{Rosen} states}, Phys. Rev. Lett.,
  \textbf{69}, 2881 (1992).

\bibitem{quantum_teleportation}
C.~H. {Bennett}, G.~{Brassard}, C.~{Crepeau}, R.~{Jozsa}, A.~{Peres} and W.~K.
  {Wootters}, \emph{{Teleporting an unknown quantum state via dual classical
  and {Einstein}-{Podolsky}-{Rosen} channels}}, Phys. Rev. Lett., \textbf{70},
  1895 (1993).

\bibitem{measurement_heisenberg_limit}
C.~M. Caves, \emph{Quantum-mechanical noise in an interferometer}, Phys. Rev.
  D, \textbf{23}, 1693 (1981).

\bibitem{quantum_information_book}
M.~A. Nielsen and I.~L. Chuang, \emph{Quantum Computation and Quantum
  Information: 10th Anniversary Edition}, Cambridge University Press, 10th ed.
  (2011).

\bibitem{quantum_algorithms}
A.~Montanaro, \emph{Quantum algorithms: an overview}, npj Quantum Inf.,
  \textbf{2}, 15023 (2016).

\bibitem{GKS_original}
V.~Gorini, A.~Kossakowski and E.~C.~G. Sudarshan, \emph{Completely positive
  dynamical semigroups of {N} level systems}, J. Math. Phys., \textbf{17}, 821
  (1976).

\bibitem{lindblad_original}
G.~Lindblad, \emph{On the generators of quantum dynamical semigroups}, Comm.
  Math. Phys., \textbf{48}, 119 (1976).

\bibitem{ChruscinskiReview}
D.~Chru{\'s}ci{\'n}ski and S.~Pascazio, \emph{A brief history of the {GKLS}
  equation}, Open Syst. Inf. Dyn., \textbf{24, 03}, 1740001 (2017).

\bibitem{using_dissipation_1}
S.~Mancini and H.~M. Wiseman, \emph{Optimal control of entanglement via quantum
  feedback}, Phys. Rev. A, \textbf{75}, 012330 (2007).

\bibitem{using_dissipation_2}
K.~Koga and N.~Yamamoto, \emph{Dissipation-induced pure {Gaussian} state},
  Phys. Rev. A, \textbf{85}, 022103 (2012).

\bibitem{Lyapunov_stationarity}
F.~Nicacio, M.~Paternostro and A.~Ferraro, \emph{Determining stationary-state
  quantum properties directly from system-environment interactions}, Phys. Rev.
  A, \textbf{94}, 052129 (2016).

\bibitem{using_hamiltonian_1}
S.~Sauer, C.~Gneiting and A.~Buchleitner, \emph{Stabilizing entanglement in the
  presence of local decay processes}, Phys. Rev. A, \textbf{89}, 022327 (2014).

\bibitem{using_hamiltonian_2}
M.~Mamaev, L.~C.~G. Govia and A.~A. Clerk, \emph{Dissipative stabilization of
  entangled cat states using a driven {B}ose-{H}ubbard dimer}, Quantum,
  \textbf{2}, 58 (2018).

\bibitem{stabilizability_geometric}
S.~Sauer, C.~Gneiting and A.~Buchleitner, \emph{Optimal coherent control to
  counteract dissipation}, Phys. Rev. Lett., \textbf{111}, 030405 (2013).

\bibitem{stabilizability_cv_systems}
{\L}.~Rudnicki and C.~Gneiting, \emph{Stabilizable {Gaussian} states}, Phys.
  Rev. A, \textbf{98}, 032120 (2018).

\bibitem{gaussian_optics}
L.~Mandel and E.~Wolf, \emph{Optical Coherence and Quantum Optics}, Cambridge
  University Press (1995).

\bibitem{gaussian_information_1}
X.-B. Wang, T.~Hiroshima, A.~Tomita and M.~Hayashi, \emph{Quantum information
  with {G}aussian states}, Phys. Rep., \textbf{448}, 1 (2008).

\bibitem{gaussian_information_2}
C.~Weedbrook, S.~Pirandola, R.~Garc\'{\i}a-Patr\'on, N.~J. Cerf, T.~C. Ralph,
  J.~H. Shapiro and S.~Lloyd, \emph{Gaussian quantum information}, Rev. Mod.
  Phys., \textbf{84}, 621 (2012).

\bibitem{controllabel_generation_of_two_mode_entangled_states}
S.-L. Ma, Z.~Li, A.-P. Fang, P.-B. Li, S.-Y. Gao and F.-L. Li,
  \emph{Controllable generation of two-mode-entangled states in two-resonator
  circuit {QED} with a single gap-tunable superconducting qubit}, Phys. Rev. A,
  \textbf{90}, 062342 (2014).

\bibitem{serafini}
A.~Serafini, \emph{Quantum Continuous Variables: A Primer of Theoretical
  Methods}, CRC Press, Taylor \& Francis Group (2017).

\bibitem{Gaussian_states_in_technologies}
S.~Tserkis and T.~C. Ralph, \emph{Quantifying entanglement in two-mode
  {G}aussian states}, Phys. Rev. A, \textbf{96}, 062338 (2017).

\bibitem{two-mode_gaussian_etc_proper_norm}
G.~Adesso, A.~Serafini and F.~Illuminati, \emph{Determination of continuous
  variable entanglement by purity measurements}, Phys. Rev. Lett., \textbf{92},
  087901 (2004).

\bibitem{Marian_2001}
P.~Marian, T.~A. Marian and H.~Scutaru, \emph{Inseparability of mixed two-mode
  {Gaussian} states generated with {SU(1,1)} interferometer}, J. Phys. A: Math.
  Gen., \textbf{34, 35}, 6969 (2001).

\bibitem{Three_qubits_can_be_entangled_in_two_inequivalent_ways}
W.~D\"ur, G.~Vidal and J.~I. Cirac, \emph{Three qubits can be entangled in two
  inequivalent ways}, Phys. Rev. A, \textbf{62}, 062314 (2000).

\bibitem{entangling_power_of_multipartite_gates}
T.~Linowski, G.~Rajchel-Mieldzio{\'{c}} and K.~{\.{Z}}yczkowski,
  \emph{Entangling power of multipartite unitary gates}, J. Phys. A: Math.
  Theor., \textbf{53}, 125303 (2020).

\bibitem{geometry_of_quantum_states}
I.~Bengtsson and K.~\.{Z}yczkowski, \emph{Geometry of Quantum States: An
  Introduction to Quantum Entanglement}, Cambridge University Press (2006).

\bibitem{two_modes_damping}
S.~Mancini, \emph{Markovian feedback to control continuous-variable
  entanglement}, Phys. Rev. A, \textbf{73}, 010304 (2006).

\bibitem{OPOs_computation}
A.~S. Parkins, E.~Solano and J.~I. Cirac, \emph{Unconditional two-mode
  squeezing of separated atomic ensembles}, Phys. Rev. Lett., \textbf{96},
  053602 (2006).

\bibitem{OPOs_entanglement}
A.~Tan, C.~Xie and K.~Peng, \emph{Bright three-color entangled state produced
  by cascaded optical parametric oscillators}, Phys. Rev. A, \textbf{85},
  013819 (2012).

\bibitem{squeezed_dissipation_experimental_1}
C.~A. Muschik, E.~S. Polzik and J.~I. Cirac, \emph{Dissipatively driven
  entanglement of two macroscopic atomic ensembles}, Phys. Rev. A, \textbf{83},
  052312 (2011).

\bibitem{squeezed_dissipation_experimental_2}
H.~Krauter, C.~A. Muschik, K.~Jensen, W.~Wasilewski, J.~M. Petersen, J.~I.
  Cirac and E.~S. Polzik, \emph{Entanglement generated by dissipation and
  steady state entanglement of two macroscopic objects}, Phys. Rev. Lett.,
  \textbf{107}, 080503 (2011).

\bibitem{OPOs_spectroscopy}
M.~Vaidyanathan, R.~C. Eckardt, V.~Dominic, L.~E. Myers and T.~P. Grayson,
  \emph{Cascaded optical parametric oscillations}, Opt. Express, \textbf{1, 2},
  49 (1997).

\bibitem{cv_systems_gaussian_states}
G.~Adesso, S.~Ragy and A.~R. Lee, \emph{Continuous variable quantum
  information: {Gaussian} states and beyond}, Open Syst. Inf. Dyn.,
  \textbf{21}, 1440001 (2014).

\bibitem{two-mode_gaussian_etc}
G.~Adesso, A.~Serafini and F.~Illuminati, \emph{Extremal entanglement and
  mixedness in continuous variable systems}, Phys. Rev. A, \textbf{70}, 022318
  (2004).

\bibitem{squeezed_metrology}
C.~M. Caves, \emph{Quantum-mechanical noise in an interferometer}, Phys. Rev.
  D, \textbf{23}, 1693 (1981).

\bibitem{squeezed_review}
R.~Schnabel, \emph{Squeezed states of light and their applications in laser
  interferometers}, Phys. Rep., \textbf{684}, 1  (2017).

\bibitem{PPT}
M.~Horodecki, P.~Horodecki and R.~Horodecki, \emph{Separability of mixed
  states: Necessary and sufficient conditions}, Phys. Lett. A, \textbf{223}, 1
  (1996).

\bibitem{PPT_cv_systems}
R.~Simon, \emph{{Peres}-{Horodecki} separability criterion for continuous
  variable systems}, Phys. Rev. Lett., \textbf{84}, 2726 (2000).

\bibitem{Gaussian_entanglement_measures}
G.~Adesso and F.~Illuminati, \emph{Gaussian measures of entanglement versus
  negativities: {Ordering} of two-mode {Gaussian} states}, Phys. Rev. A,
  \textbf{72}, 032334 (2005).

\bibitem{Bures_distance}
P.~Marian and T.~A. Marian, \emph{Bures distance as a measure of entanglement
  for symmetric two-mode {Gaussian} states}, Phys. Rev. A, \textbf{77}, 062319
  (2008).

\bibitem{entropies_tsallis}
C.~Tsallis, \emph{Possible generalization of {Boltzmann-Gibbs} statistics}, J.
  Stat. Phys., \textbf{52}, 479 (1988).

\bibitem{entropies_renyi}
A.~R{\'e}nyi, \emph{On measures of entropy and information}, \emph{Proceedings
  of the Fourth Berkeley Symposium on Mathematical Statistics and Probability,
  Volume 1: Contributions to the Theory of Statistics}, 547--561, University of
  California Press, Berkeley, Calif. (1961).

\bibitem{Note1}
Note that it is necessary for an entangled state to have $c_+$ and $c_-$ with
  opposite signs. This can be shown by manipulating eq. (\ref
  {constraints_general_log_neg}) and taking into account the first of the
  constrains (\ref {constraints_general_h}).

\bibitem{Note2}
Some readers may be familiar with the fact that in finite-dimensional systems
  all states that are {sufficiently close} to the maximally mixed state are
  separable \cite {separable_ball}. We stress that this fact does not extend to
  continuous variable systems, and so there is no inconsistency with our
  findings connecting the amount of entanglement to the amount of mixedness in
  stabilizable states.

\bibitem{ising_machine_1}
Y.~Inui and Y.~Yamamoto, \emph{Steady-state squeezing and entanglement in a
  dissipatively coupled {NOPO} network} (2019), arXiv:1906.12044.

\bibitem{Vogel05}
E.~Shchukin and W.~Vogel, \emph{Inseparability criteria for continuous
  bipartite quantum states}, Phys. Rev. Lett., \textbf{95}, 230502 (2005).

\bibitem{separable_ball}
L.~Gurvits and H.~Barnum, \emph{Largest separable balls around the maximally
  mixed bipartite quantum state}, Phys. Rev. A, \textbf{66}, 062311 (2002).

\end{thebibliography}
\bibliographystyle{custombib}

\end{document}